\documentclass[11pt]{amsart}
\usepackage[a4paper,margin=1.0in]{geometry}
\usepackage{times}
\usepackage[numbers,sort&compress]{natbib}
\usepackage{amsmath,amssymb,amsfonts,amsthm,mathtools,bm}
\usepackage{booktabs,array,multirow,tabularx}
\usepackage{graphicx,subcaption}
\usepackage{algorithm,algpseudocode}
\usepackage{microtype,url,hyperref,cleveref}

\newcommand{\SPD}{\mathbb S_{++}}

\newcommand{\R}{\mathbb R}
\newcommand{\E}{\mathbb E}
\newcommand{\Prob}{\mathbb P}

\newcommand{\tr}{\operatorname{tr}}
\newcommand{\diag}{\operatorname{diag}}
\newcommand{\AIRM}{\mathrm{AI}}
\newcommand{\Rhat}{\widehat R}
\newcommand{\Revidence}{R_e}

\newcommand{\safegraphic}[2][\linewidth]{%
  \IfFileExists{#2}{\includegraphics[width=#1]{#2}}{%
  \fbox{\parbox[c][0.17\textheight][c]{0.90\linewidth}{\centering
  \textbf{Figure placeholder}\\[2mm]\texttt{\detokenize{#2}}}}}}

\theoremstyle{plain}
\newtheorem{lemma}{Lemma}
\newtheorem{proposition}{Proposition}
\newtheorem{corollary}{Corollary}
\newtheorem{theorem}{Theorem}
\theoremstyle{definition}
\newtheorem{remark}{Remark}

\title{Bayesian Matrix-Valued Graphs for Context-Dependent Multivariate Relationships}
\author{Papri Dey}
\date{September 2026}

\begin{document}

\begin{abstract}
Many scientific graphs attach several variables to each node, so a single scalar edge weight cannot describe direction-dependent interactions. We model each edge by a symmetric positive-definite (SPD) matrix and infer a posterior over matrix-valued graph geometries, which we call the Bayesian matrix-valued graph (BMVG). We ask how these interactions reconfigure across contexts: how large the change is and which multivariate directions strengthen or weaken. The geodesic distance induced by the affine-invariant Riemannian metric
(AIRM) quantifies deformation magnitude and generalized eigenvalues resolve its signed directions.Against fused graphical lasso, Bayesian multiple-GGM, and common principal
components, BMVG is competitive on global precision recovery while retaining
identifiable matrix-valued edge structure and accurately recovering edge-level
deformation directions. In controlled known-truth experiments, it resolves structural change with increasing sample size, including orientation changes that leave ordinary eigenvalues unchanged. In one year of Bay Area weather data, the geometry of 12-hour change reconfigures spatial coupling about as much as whole seasons differ. In TCGA-BRCA, estrogen-receptor (ER)-associated reconfiguration concentrates on specific gene-module pairs and persists under graph-scaffold sparsification and removal of subgroup mean differences. These results establish posterior matrix-valued edge geometry as a unified framework for quantifying and interpreting context-dependent multivariate reconfiguration.
\end{abstract}
\maketitle

\section{Introduction}
\label{sec:intro}

A scalar graph edge records only how strongly two nodes are connected. This
is inadequate when each node carries several variables. A weather station
reports temperature, relative humidity, wind, and pressure; a biological
module comprises several genes. The relationship between two such nodes can be
strong along one combination of variables and weak along another, and it can
rotate when the scientific context changes, a reconfiguration that no
single scalar weight can express.

We therefore represent each edge $e$ of a fixed graph $G=(V,E)$ by a symmetric
positive-definite (SPD) matrix $W_e\in\mathcal{S}^{d}_{++}$, i.e.\ $x^{\top}
W_e x>0$ for every nonzero $x\in\mathbb{R}^{d}$. Inference targets the
posterior of $W=\{W_e\}_{e\in E}$, a distribution over matrix-valued graph
geometries rather than over scalar edge strengths.

Our aim is not to fit a graph in each context, but to characterize how the
multivariate relationships it encodes change across contexts
: different weather regimes, atmospheric state versus short-horizon innovation, or distinct
tumor subgroups. We distinguish three complementary aspects of change:
(1)~\emph{magnitude}, how far the corresponding SPD edge matrices move;
(2)~\emph{uncertainty}, whether that separation exceeds the finite-sample
variation expected from independently fitted posteriors under no structural
change; and (3)~\emph{direction}, which multivariate combinations strengthen
or weaken between contexts. Together these separate genuine, interpretable
reconfiguration from ordinary estimation variability.

The ingredients are individually classical: affine-invariant geometry on SPD
matrices~\cite{bhatia2007,pennec2006}, and Gaussian graphical models with cross-context
comparison of covariance or precision structure~\cite{lauritzen1996,flury1984,danaher2014,peterson2015}.
We claim no new SPD metric or MCMC method. Our contribution is their
graph-specific synthesis into a posterior over matrix-valued edges, and 
the demonstration that this representation directly characterizes
edge-level directional reconfiguration and can recover it substantially more
accurately than the comparison methods considered here. ConeMALA~\cite{sampler} provides the geometry-aware MCMC inference engine used
to sample these posteriors. Here we build on the resulting Bayesian
matrix-valued graph (BMVG) posterior to study context-dependent deformation,
uncertainty, and generalized directions of change.

Our contributions are as follows. First, we use the AIRM-induced geodesic
distance as an intrinsic edge effect size and a generalized eigenproblem to
resolve that distance into signed strengthening and weakening directions.
Second, we introduce a same-truth reference that compares a planted change
against the separation between posteriors fitted to two independent datasets
from the same graph, so that finite-sample posterior spread is not mistaken for
structural change. Third, on a $120$-condition controlled grid we compare BMVG
against fused graphical lasso~\cite{danaher2014}, Bayesian multiple-GGM~\cite{peterson2015},
and common principal components~\cite{flury1984}, each in its native inference: BMVG is competitive on global precision recovery while retaining identifiable
matrix-valued edge structure and accurately recovering edge-level deformation
directions. Under the stated conditions, the matrix-valued edge parameterization is
identifiable: the edge blocks are uniquely determined by the induced precision
matrix. This permits context-dependent changes to be attributed to specific
edges, whose deformation magnitude and generalized directions can then be
estimated from the posterior.
Finally, we evaluate the framework in controlled, meteorological, and
biological settings, where it resolves orientation changes invisible to
ordinary eigenvalues, finds that $12$-hour weather innovation reconfigures
spatial geometry about as much as whole seasons differ, and localizes
estrogen-receptor-associated change in TCGA-BRCA to specific gene-module pairs
that persist under scaffold and mean-shift perturbation.

\section{Bayesian matrix-valued graph model}
\label{sec:model}
Let $G=(V,E)$ be a fixed undirected graph with $m=|V|$ nodes, and let each node
contain a $d$-vector.  Choose an arbitrary orientation for each edge and let
$b_e\in\R^m$ be its incidence vector.  Define
\[
B_e=b_e\otimes I_d\in\R^{md\times d},
\]
where $\otimes$ is the Kronecker product.  Reversing the orientation changes
$b_e$ to $-b_e$ but leaves the formulas below unchanged.

For edge matrices $W=(W_e)_{e\in E}$, define the block Laplacian and precision
matrix
\begin{equation}
L(W)=\sum_{e\in E}B_eW_eB_e^\top,
\qquad Q(W)=L(W)+R,
\qquad R\succ0.
\label{eq:block-laplacian}
\end{equation}
Here $R$ is a fixed positive-definite stabilizer.  If
$y=(y_1^\top,\ldots,y_m^\top)^\top\in\R^{md}$, then
\begin{equation}
y^\top L(W)y
=\sum_{e=(i,j)\in E}(y_i-y_j)^\top W_e(y_i-y_j).
\label{eq:laplacian-energy}
\end{equation}
Thus $W_e$ describes how disagreement between the two endpoint vectors is
penalized in different directions.  The inverse $Q(W)^{-1}$ is the covariance
matrix under the Gaussian working model, while $Q(W)$ is the precision matrix.

The notation above assumes that the endpoint vectors use the same coordinates.
When two nodes contain different variables, this direct subtraction is not
appropriate.  \Cref{app:fixedmap} gives a fixed-map extension that
replaces $y_i-y_j$ by
$A_{e,i}y_i-A_{e,j}y_j$ in an edge-specific shared space.  The likelihood,
edge score, AIRM comparison, and ConeMALA update keep the same form after this
replacement.  The TCGA-BRCA analysis uses the one-dimensional specialization of this
extension.

Observations are centered according to the experiment-specific preprocessing
described in \Cref{sec:experiments}, and in the corresponding appendices.  Therefore, the Gaussian likelihood is written with zero mean.
For centered observations $Y_1,\ldots,Y_n\in\R^{md}$, let
\[
S=\frac1n\sum_{r=1}^nY_rY_r^\top.
\]
Conditional on the
edge matrices $W=\{W_e:e\in E\}$,
we use the Gaussian working model
\begin{equation}
Y_r\mid W
\sim
N\!\left(0,Q(W)^{-1}\right),
\qquad
\ell(W)
=
\frac n2
\left\{
\log\det Q(W)
-
\tr\!\bigl(SQ(W)\bigr)
\right\}
+C.
\label{eq:likelihood}
\end{equation}
Here $Q(W)$ is the graph precision matrix defined in
\Cref{eq:block-laplacian}; the precision is the inverse covariance, so
$Q(W)^{-1}$ describes the covariance implied by the matrix-valued graph.
The zero mean is appropriate because the observations have been centered.
The matrix $S
=
\frac1n
\sum_{r=1}^n
Y_rY_r^\top$
is the empirical second-moment matrix, $\det(\cdot)$ denotes the determinant,
$\tr(\cdot)$ denotes the matrix trace, and $C$ collects terms that do not
depend on $W$.

We use this likelihood to connect the observed multivariate variation to the
unknown edge matrices: values of $W$ are favored when the corresponding
precision $Q(W)$ explains the empirical dependence summarized by $S$.
The Gaussian distribution is therefore a \emph{working model} for
second-order dependence.  We do not require the scientific observations to be
exactly Gaussian, and the inferred edges are associational, not causal.

The edgewise likelihood gradient is
\begin{equation}
\nabla_{W_e}\ell(W)
=\frac n2 B_e^\top\{Q(W)^{-1}-S\}B_e.
\label{eq:edge-score}
\end{equation}
Although the SPD constraints are edgewise, all edge gradients depend on the
same global inverse $Q(W)^{-1}$.  ConeMALA writes its proposal relative to the
product affine-invariant Riemannian volume.  After converting the Wishart
posterior from product Lebesgue measure, the corresponding intrinsic negative
log density is, up to a constant,
\begin{equation}
\Phi(W)=
\frac n2\{\tr(SQ(W))-\log\det Q(W)\}
-\frac{\nu}{2}\sum_{e\in E}\log\det W_e
+\frac12\sum_{e\in E}\tr(\Psi^{-1}W_e).
\label{eq:intrinsic-potential}
\end{equation}
Independent Wishart priors
$W_e\sim\mathcal W_d(\nu,\Psi)$ keep every edge SPD and they are shared across the contexts being compared. Posterior draws are obtained with ConeMALA.
We run multiple chains and monitor the rank-normalized split-$\Rhat$ diagnostic
(chain agreement) and effective sample size, ESS (the number of approximately
independent draws represented by the correlated chain).  Exact settings and
all diagnostic gates are in \Cref{app:inference}.

\section{Posterior geometry and structural reconfiguration}
\label{sec:geometry}
Context- or group-specific posterior inference is standard in Bayesian
multigroup graphical modeling; see, for example, \cite{peterson2015}.
For context $c$, let $\Pi_c(W)=p(W\mid\mathcal D_c)$ 
denote the posterior based on observations $\mathcal D_c$.  A context may be a weather regime, season,
biological group, controlled perturbation, or a different representation of
the same data.

For $A,B\in\mathbb S_{++}^d$, the geodesic distance induced by the
AIRM is
\begin{equation}
d_{\AIRM}(A,B)=\left\|\log(A^{-1/2}BA^{-1/2})\right\|_F,
\label{eq:airm}
\end{equation}
where $\|\cdot\|_F$ is the Frobenius norm.  It is invariant to a common
invertible change of coordinates.  For edge $e$ and contexts $a,b$ we use
\begin{equation}
D_e^{a,b}=d_{\AIRM}(W_e^a,W_e^b)
\label{eq:posterior-distance}
\end{equation}
as the posterior edge effect size.  We report distances between posterior-mean
blocks separately from draw-pair distances because the latter also contain
posterior uncertainty from both fits.

AIRM says how much an edge changed but not how.  For $A,B\in\SPD^d$, solve
\begin{equation}
Bv_k=\lambda_kAv_k,
\qquad \eta_k=\log\lambda_k.
\label{eq:generalized-eigenproblem}
\end{equation}
The $\lambda_k>0$ are the generalized eigenvalues of the pair $(B,A)$.

\begin{proposition}
\label{prop:distance-spectrum}
For the generalized eigenvalues in \eqref{eq:generalized-eigenproblem},
\begin{equation}
d_{\AIRM}(A,B)^2=\sum_{k=1}^d(\log\lambda_k)^2
=\sum_{k=1}^d\eta_k^2.
\label{eq:distance-spectrum}
\end{equation}
\end{proposition}
This eigenvalue form of the affine-invariant distance is classical~\cite{bhatia2007,FM03}. A positive $\eta_k$ means that $B$ is stronger relative to $A$ along the corresponding generalized direction; a negative value means weaker.  Different
$\eta_k$ may have different signs, so one matrix edge can strengthen and weaken
simultaneously in different directions.

For an edge $e$, let
$
D_{e,c}:=
d_{\AIRM}\!\left(W_e^{(0)},W_e^{(c)}\right)
$ denote the AIRM distance between an independently paired draw from the
baseline posterior and a draw from the posterior under candidate context
$c$.  This is a draw-level measure of the size of the inferred edge change.
Likewise, let $D_{e,\mathrm{null}}
=
d_{\AIRM}\!\left(W_e^{(0)},W_e^{(0')}\right),$ 
where the two posteriors are fitted to independent datasets generated from
the same underlying edge matrix.

Even when there is no true structural change, two continuous posteriors fitted
to independent finite data sets produce positive draw-pair distances.
Therefore $\Prob(D_{e,c}>0)$ is not a useful equality test.  In the known-truth structural-resolution study, we instead define
\begin{equation}
\Revidence(c)
=
\Prob\!\left(
D_{e,c}>D_{e,\mathrm{null}}
\mid\mathcal D
\right).
\label{eq:null-evidence}
\end{equation}
Thus $D_{e,c}$ measures the magnitude of the candidate change, whereas
$\Revidence(c)$ measures how often that change is larger than the
finite-sample separation observed under no change.  The prespecified rule
$\Revidence(c)\ge0.95$ declares a change detectable when its posterior
draw-pair distance exceeds the corresponding no-change distance with at least
$95\%$ probability.  This threshold is used only to calibrate the controlled
known-truth experiment and is not proposed as a universal scientific
significance rule.

\section{Experiments}
\label{sec:experiments}
The studies have distinct roles. The known-truth structural-resolution study
measures finite-sample resolution and calibrates the separation expected under
no change. The Meteostat analysis examines how spatial matrix-valued graph
geometry varies across physical and observational contexts. The TCGA-BRCA
population analysis studies ER-associated module-pair reconfiguration. All
main posterior fits use multiple chains and prespecified $\Rhat$/ESS diagnostic
gates; sampler settings and diagnostic criteria are reported in
\Cref{app:inference}.
\subsection{Known-truth structural resolution}
\label{sec:expA}

We use five modules of dimension $d=3$ connected by six fixed edges.  For each
of five independent generating seeds, a fresh baseline truth is sampled and
then held fixed while controlled perturbations are applied.  We use
$n\in\{200,500,1000\}$ and plant changes with target AIRM distances
$D^\star\in\{0.5,0.75,1,1.5\}$.  A sparse perturbation changes one
generalized direction.  An orientation perturbation uses an orthogonal
congruence $RAR^\top$ and is tuned to the same AIRM magnitude while preserving
trace, determinant, and ordinary eigenvalues.  The latter is invisible to any
summary based only on ordinary eigenvalues.

All 165 required posterior fits pass the diagnostic gate.  Under exactly
unchanged truth, the mean draw-pair distance decreases from $1.052$ at $n=200$
to $0.636$ at $n=500$ and $0.449$ at $n=1000$, closely following an empirical
$n^{-1/2}$ scale. With $R_e\ge0.95$ taken as a positive declaration of structural change,
specificity is defined as
\[
\operatorname{Specificity}
=
\frac{\mathrm{TN}}{\mathrm{TN}+\mathrm{FP}},
\]
where $\mathrm{TN}$ is the number of truly unchanged edges with
$R_e<0.95$ and $\mathrm{FP}$ is the number of truly unchanged edges
incorrectly declared changed, i.e., with $R_e\ge0.95$.
The resulting specificity is $1.00$ in every reported sparse and
orientation condition. At planted effect $D^\star=1$, the orientation change is detected in $0/5$, $4/5$, and $5/5$ independent data-generating replicates at $n=200$, $500$, and $1000$, respectively.  The corresponding sparse-change detection rates are $0/5$, $1/5$, and $3/5$. Thus equal AIRM magnitude need not imply equal statistical difficulty.

\begin{figure}[htbp]
\centering
\begin{subfigure}[t]{0.32\linewidth}\centering
\safegraphic{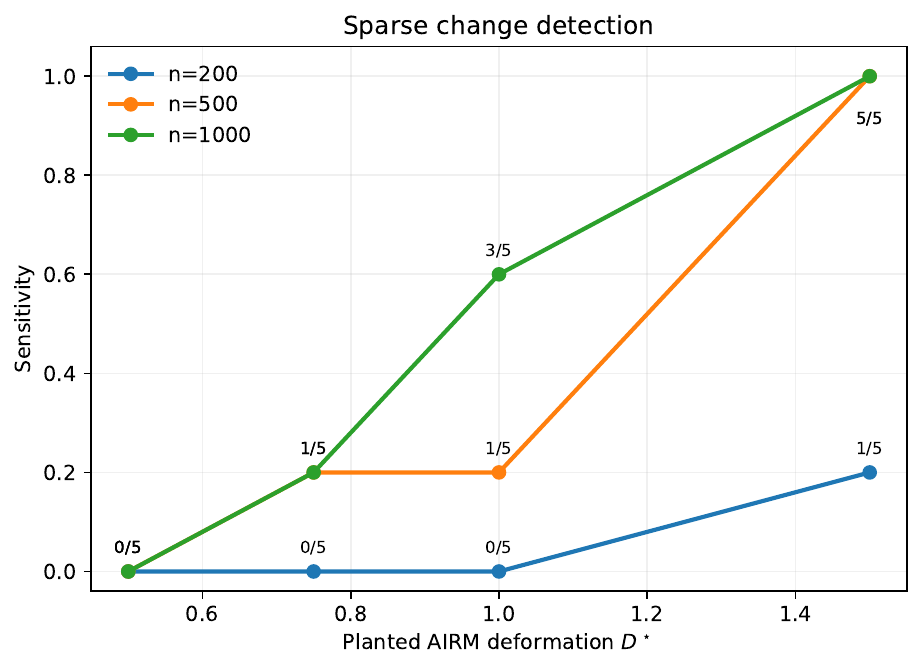}
\caption{Sparse change.}
\end{subfigure}\hfill
\begin{subfigure}[t]{0.32\linewidth}\centering
\safegraphic{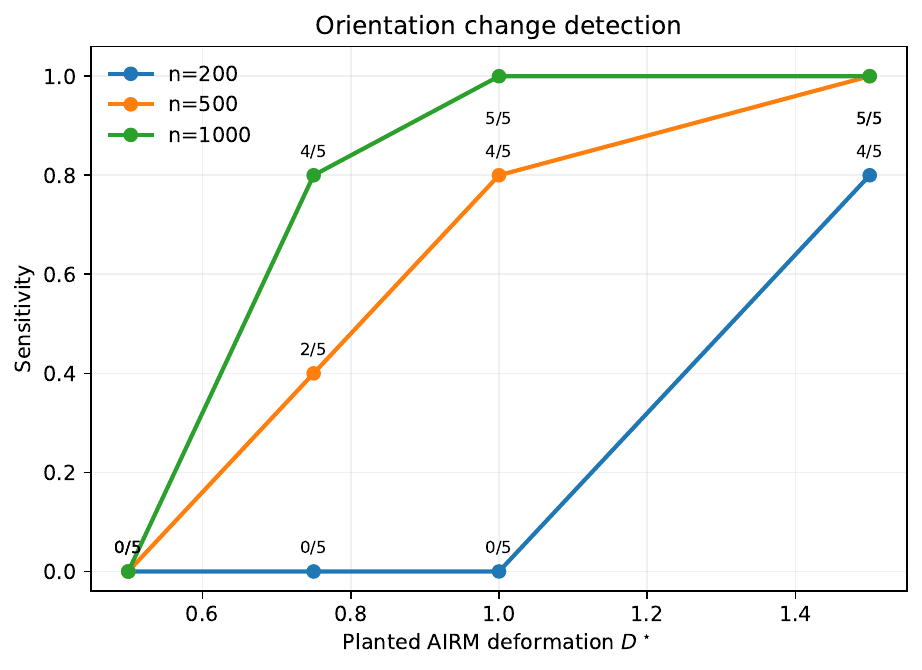}
\caption{Orientation change.}
\end{subfigure}\hfill
\begin{subfigure}[t]{0.32\linewidth}\centering
\safegraphic{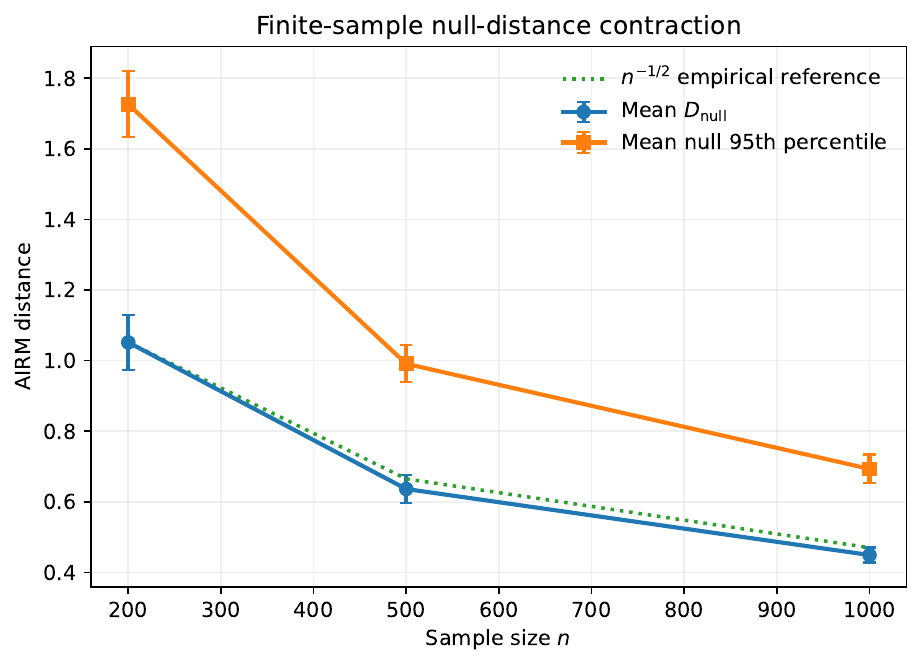}
\caption{Null contraction.}
\end{subfigure}
\caption{
Panels (A-B) show the sensitivity of the prespecified
$R_e\ge0.95$ rule across five independently generated synthetic
data replicates.  Panel (C) shows the positive distance floor produced by two
independently fitted no-change posteriors and its contraction with sample size.}
\label{fig:expA_detection}
\end{figure}

The generalized modes recover how the edge changed.  At $D^\star=1$, mean
direction alignment improves from $0.853$ to $0.956$ for sparse deformation
and from $0.928$ to $0.978$ for orientation deformation as $n$ grows from 200
to 1000.  Sign-recovery probability is essentially one at the largest sample
size.  Full detection tables, magnitude recovery, and direction/sign plots are
in \Cref{app:expA}.

\subsection{Comparison with native multigroup estimators}
\label{sec:expA_native_comparison}

To assess whether the gains above arise from the matrix-valued graph
representation rather than only from the sampling procedure, we compare the
BMVG method with three established estimators, each fitted using its native
inference mechanism. DWW14 fused graphical lasso (FGL)
\citep{danaher2014} estimates two unrestricted $15\times15$ precision
matrices, with sparsity and fusion penalties selected by validation Gaussian
negative log-likelihood (NLL), without access to the planted truth. PSV15
\citep{peterson2015} is a Bayesian multiple-GGM model fitted with its native
G-Wishart MCMC; we evaluate both its posterior-mean precision matrices and its
differential edge-support probabilities. Flury84 common principal components
(CPC) \citep{flury1984} is fitted by profiled Gaussian maximum likelihood to
the three-dimensional contrast on the planted changed edge, providing a direct
benchmark for covariance and directional recovery.

All methods use the same training observations, while fresh common
validation/test samples are used only for tuning and evaluation. The full
comparison covers five independent replicates,
$n\in\{200,500,1000\}$, $D^\star\in\{0.5,0.75,1,1.5\}$, and both perturbation
families, giving $5\times3\times4\times2=120$
matched conditions. Thus the study compares complete statistical procedures,
rather than attributing differences solely to parameterization.

For the candidate context, $Q_1^\star$ denotes the data-generating precision
matrix and $\widehat Q_1$ its estimate; the corresponding precision change is
\[\Delta Q^\star=Q_1^\star-Q_0^\star,
\
\widehat{\Delta Q}=\widehat Q_1-\widehat Q_0.\]
For the changed edge $e=(i,j)$, let $\delta_e(Y)=B_e^\top Y=Y_i-Y_j$ 
denote its $d$-dimensional contrast, with context-specific covariance
\[
\Sigma_{e,c}
=
\operatorname{Cov}\{\delta_e(Y)\}
=
B_e^\top Q_c^{-1}B_e,
\qquad
\Delta\Sigma_e=\Sigma_{e,1}-\Sigma_{e,0}.
\]
Here $\Sigma_{e,c}^\star$ and $\Delta\Sigma_e^\star$ denote the corresponding
data-generating quantities. ``Axis error'' is the absolute principal angle,
in degrees, between the estimated and planted leading deformation directions.
NLL denotes Gaussian negative log-likelihood, so lower values indicate better
predictive fit.

\begin{table*}[htbp]
\centering
\scriptsize
\caption{
Panels A-B report mean \(\pm\) sample standard deviation across five
independent replicates at \(D^\star=1\) and \(n=1000\); lower is better.
Panel~C reports pairwise win counts over the full \(120\)-condition grid:
an entry \(a/120\) means that BMVG has strictly smaller error (or NLL) than
the indicated comparator in \(a\) of the \(120\) matched conditions.}
\label{tab:native_method_comparison}
\setlength{\tabcolsep}{4.2pt}
\begin{tabular}{llccc}
\toprule
\multicolumn{5}{l}{\textbf{Panel A: global precision recovery and predictive fit}}\\
Family & Method &
\(\|\widehat Q_1-Q_1^\star\|_F/\|Q_1^\star\|_F\) &
\(\|\widehat{\Delta Q}-\Delta Q^\star\|_F/\|\Delta Q^\star\|_F\) &
Test NLL\\
\midrule
Sparse & BMVG & \textbf{0.066 \(\pm\) 0.014} &
\textbf{0.303 \(\pm\) 0.068} & \textbf{15.752 \(\pm\) 0.299} \\
 & DWW14 FGL & 0.168 \(\pm\) 0.037 & 0.743 \(\pm\) 0.066 &
15.782 \(\pm\) 0.297 \\
 & PSV15 multiGGM & 0.086 \(\pm\) 0.013 & 0.416 \(\pm\) 0.076 &
15.787 \(\pm\) 0.295 \\
\addlinespace[2pt]
Orientation & BMVG & \textbf{0.075 \(\pm\) 0.010} &
0.829 \(\pm\) 0.212 & \textbf{15.859 \(\pm\) 0.278} \\
 & DWW14 FGL & 0.085 \(\pm\) 0.005 &
\textbf{0.686 \(\pm\) 0.099} & 15.883 \(\pm\) 0.280 \\
 & PSV15 multiGGM & 0.101 \(\pm\) 0.009 & 1.074 \(\pm\) 0.221 &
15.892 \(\pm\) 0.274 \\
\midrule
\multicolumn{5}{l}{\textbf{Panel B: changed-edge multivariate geometry}}\\
Family & Method &
\(\|\widehat{\Delta\Sigma}_e-\Delta\Sigma_e^\star\|_F/
 \|\Delta\Sigma_e^\star\|_F\) &
Axis error (deg.) &
Edge NLL\\
\midrule
Sparse & BMVG & \textbf{0.210 \(\pm\) 0.065} &
\textbf{0.70 \(\pm\) 0.58} & \textbf{3.208 \(\pm\) 0.167} \\
 & Flury84 CPC & 0.642 \(\pm\) 0.141 &
20.07 \(\pm\) 19.14 & 3.214 \(\pm\) 0.170 \\
\addlinespace[2pt]
Orientation & BMVG & \textbf{0.358 \(\pm\) 0.088} &
\textbf{4.05 \(\pm\) 4.07} & \textbf{3.643 \(\pm\) 0.157} \\
 & Flury84 CPC & 1.033 \(\pm\) 0.019 &
25.29 \(\pm\) 5.90 & 3.662 \(\pm\) 0.157 \\
\midrule
\multicolumn{5}{l}{\textbf{Panel C: matched conditions in which BMVG performs better}}\\
Comparator & & Context/candidate error & Change error &
NLL \(\quad/\quad\) Direction\\
\midrule
DWW14 FGL & & 109/120 & 62/120 & 120/120 \(\quad/\quad\) -- \\
PSV15 multiGGM & & 120/120 & 120/120 & 120/120 \(\quad/\quad\) -- \\
Flury84 CPC & & 98/120 & 111/120 & 112/120 \(\quad/\quad\) 115/120 \\
\bottomrule
\end{tabular}

\vspace{2pt}
\parbox{0.98\linewidth}{\footnotesize
For DWW14 and PSV15, ``context/candidate error'' is the relative Frobenius
error of \(Q_1\), and ``change error'' is the relative Frobenius error of
\(\Delta Q\). For the Flury84 CPC row in Panel~C, the corresponding quantities
are changed-edge covariance error and covariance-change error; NLL is computed
on the edge contrast and ``Direction'' denotes leading-axis error. The
orientation-family \(\Delta Q\) result is intentionally not uniformly
favorable to BMVG: FGL can estimate the aggregate precision change
\(\Delta Q\) more accurately, whereas BMVG is substantially more accurate
for recovering the edge-level deformation direction.}
\end{table*}

\begin{figure*}[htbp]
\centering
\begin{minipage}[t]{0.325\linewidth}
\centering
\includegraphics[width=\linewidth]{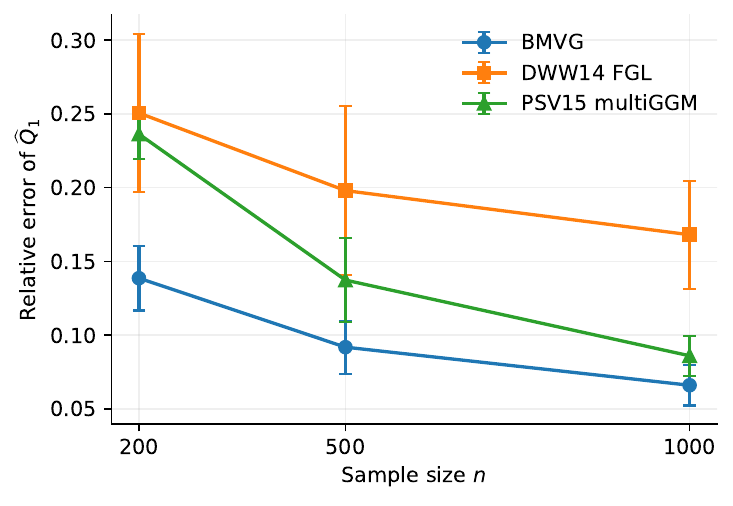}\\[-2pt]
\textbf{(a)} Sparse perturbation
\end{minipage}\hfill
\begin{minipage}[t]{0.325\linewidth}
\centering
\includegraphics[width=\linewidth]{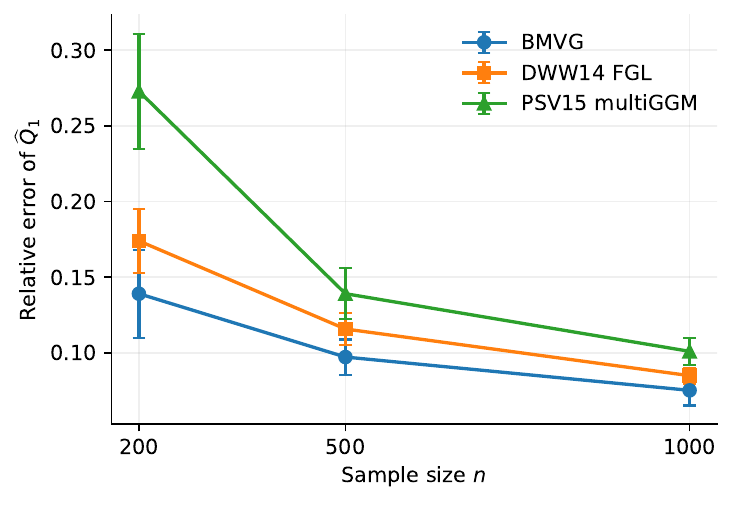}\\[-2pt]
\textbf{(b)} Orientation perturbation
\end{minipage}\hfill
\begin{minipage}[t]{0.325\linewidth}
\centering
\includegraphics[width=\linewidth]{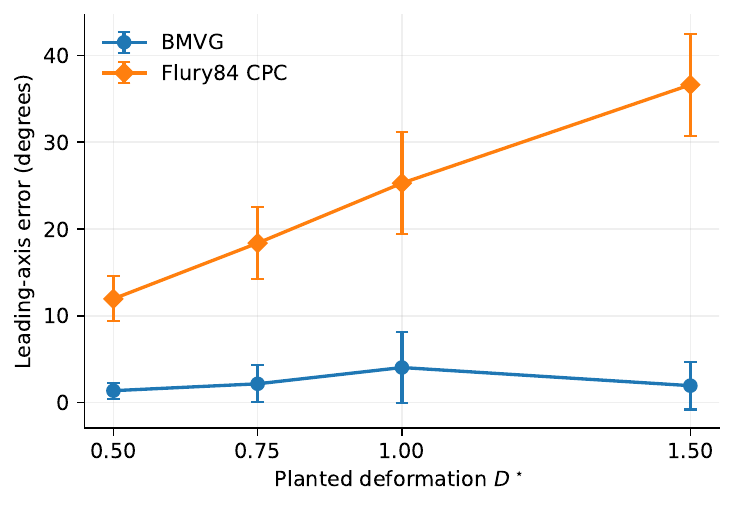}\\[-2pt]
\textbf{(c)} Directional recovery
\end{minipage}
\caption{
\footnotesize
Panels (a)-(b) show the relative error of the candidate-context precision \(Q_1\) at
\(D^\star=1\), averaged over five replicates. The BMVG estimator improves
steadily with \(n\) and is most accurate for both perturbation families;
PSV15 is particularly competitive for sparse changes, while FGL is closer
under orientation changes. (c) Mean absolute leading-axis error for the
orientation family at \(n=1000\). The BMVG matrix-valued edge tracks the
planted rotation across deformation magnitudes, whereas the common-basis
restriction of CPC yields substantially larger directional error. Error bars
show one sample standard deviation across replicates.}
\label{fig:native_method_main}
\end{figure*}

At \(D^\star=1\) and \(n=1000\), BMVG has the smallest
candidate-precision error for both perturbation families
(\Cref{tab:native_method_comparison}, Panel~A), while FGL is more accurate
for the global \(\Delta Q\) error in the orientation case. Thus BMVG is not
uniformly superior for every global difference summary; its principal gain is
in recovering the local matrix-valued deformation. 
Together, Panels A-C show that BMVG retains competitive global precision and
predictive recovery while providing substantially more accurate edge-level
multivariate deformation geometry. The corresponding precision-recovery
trends are shown in \Cref{fig:native_method_main}. Full \(\Delta Q\),
localization, PSV15 support-probability, and MCMC-stability results are
reported in \Cref{app:native_comparison_details}.

\subsection{Context-dependent spatial geometry in Meteostat}
\label{sec:meteostat}

We use one year (2023) from eight Bay Area weather stations
\citep{meteostat}. Each station has $d=4$ variables: temperature, relative
humidity, wind speed, and pressure. Hourly observations are aggregated to
3-hour means, giving 2920 multivariate spatial fields. To construct the
geographic scaffold, each station is connected to its three geographically
nearest stations, and the resulting neighbor relations are symmetrized.
After duplicate reciprocal connections are merged, this gives a fixed graph
with 16 undirected edges. All preprocessing parameters are estimated from the
training period only.

The experiment asks whether the spatial relationships describing the
atmospheric state are the same as those describing how that state changes over
the following 12 hours. To separate these roles, we use three global
representations. 
Let $X_t\in\R^{8\times 4}$ denote the physical 3-hour mean meteorological
field.  After a variable-wise support transform $g$, the \emph{raw}
representation is $Y_{t,iv}
=
\frac{
g_v(X_{t,iv})-\mu_{iv}^{\rm tr}
}{
\sigma_{iv}^{\rm tr}
},$ 
where $\mu_{iv}^{\rm tr}$ and $\sigma_{iv}^{\rm tr}$ are the empirical
training-period mean and sample standard deviation, respectively, of the
transformed variable $v$ at station $i$. The transform $g$
leaves temperature and pressure unchanged, applies a logit transform to
relative humidity, and a scaled inverse-softplus transform to wind speed.

The \emph{anomaly} representation is $A_t=Y_t-\bar Y_{h(t)}^{\rm tr},$
where $\bar Y_{h(t)}^{\rm tr}$ is the mean standardized field, computed over
training days at the same 3-hour time-of-day position as observation $t$, thereby
removing the regular diurnal component.  Finally, the 12-hour
\emph{innovation} $I_t=Y_{t+12{\rm h}}-Y_t$
describes the change in the standardized meteorological field over the
following 12 hours.
Thus, raw and anomaly representations describe spatial dependence in the
atmospheric state, whereas the innovation representation describes spatial
dependence in its short-horizon evolution. We additionally fit four
training-derived anomaly regimes and use four seasons as retrospective
descriptive contexts. Exact transformation constants and
training-only preprocessing details are given in \Cref{app:meteostat}.

\Cref{fig:meteostat_context_geometry} addresses the first question: how much does
the inferred spatial coupling change when the environmental context or
representation changes? Using \eqref{eq:posterior-distance},
we summarize the overall graph separation by averaging over the fixed edge set,
\begin{equation}
\bar D_{a,b}
=
\frac{1}{|E|}
\sum_{e\in E}
D_e^{a,b}
=
\frac{1}{|E|}
\sum_{e\in E}
d_{\rm AI}\!\left(
\bar W_e^{(a)},\bar W_e^{(b)}
\right).
\label{eq:mean-edge-airm}
\end{equation}
Using \eqref{eq:mean-edge-airm}, the mean edgewise AIRM separations are
\[
\bar D_{\rm raw,anomaly}=0.887,\qquad
\bar D_{\rm raw,innovation}=3.184,\qquad
\bar D_{\rm anomaly,innovation}=3.576.
\]
As reference comparisons, we also evaluate anomaly-based posterior graphs
across atmospheric regimes and seasons.  Averaging the edgewise AIRM
distances as in \eqref{eq:mean-edge-airm} gives mean separations of $2.800$
between regimes and $3.511$ between seasons.  The
anomaly-innovation separation, $3.576$, is therefore comparable to the
substantial graph reconfiguration observed across seasons rather than to a
small perturbation of the anomaly geometry.

This contrast is especially clear when compared with
$\bar D_{\rm raw,anomaly}=0.887$.  Removing the mean diurnal cycle changes
the learned graph relatively little, whereas replacing the atmospheric state
by its 12-hour change produces a much larger reconfiguration.  Thus, stations
that are strongly related in their contemporaneous meteorological state need
not exhibit the same multivariate relationship in their short-horizon
evolution.  For environmental monitoring, this suggests that a graph
describing spatial state similarity need not also describe how atmospheric
changes evolve across the station network. Seasonal dependence also varies substantially.  The largest mean seasonal
separation is autumn-spring ($4.168$), while spring-summer is the smallest
($2.347$).  \Cref{fig:meteostat_context_geometry}(b) therefore reflects
changes in the fitted multivariate dependence structure across stations, not
merely changes in marginal weather values.

\begin{figure}[htbp]
\centering
\begin{subfigure}[t]{0.5\linewidth}\centering
\safegraphic{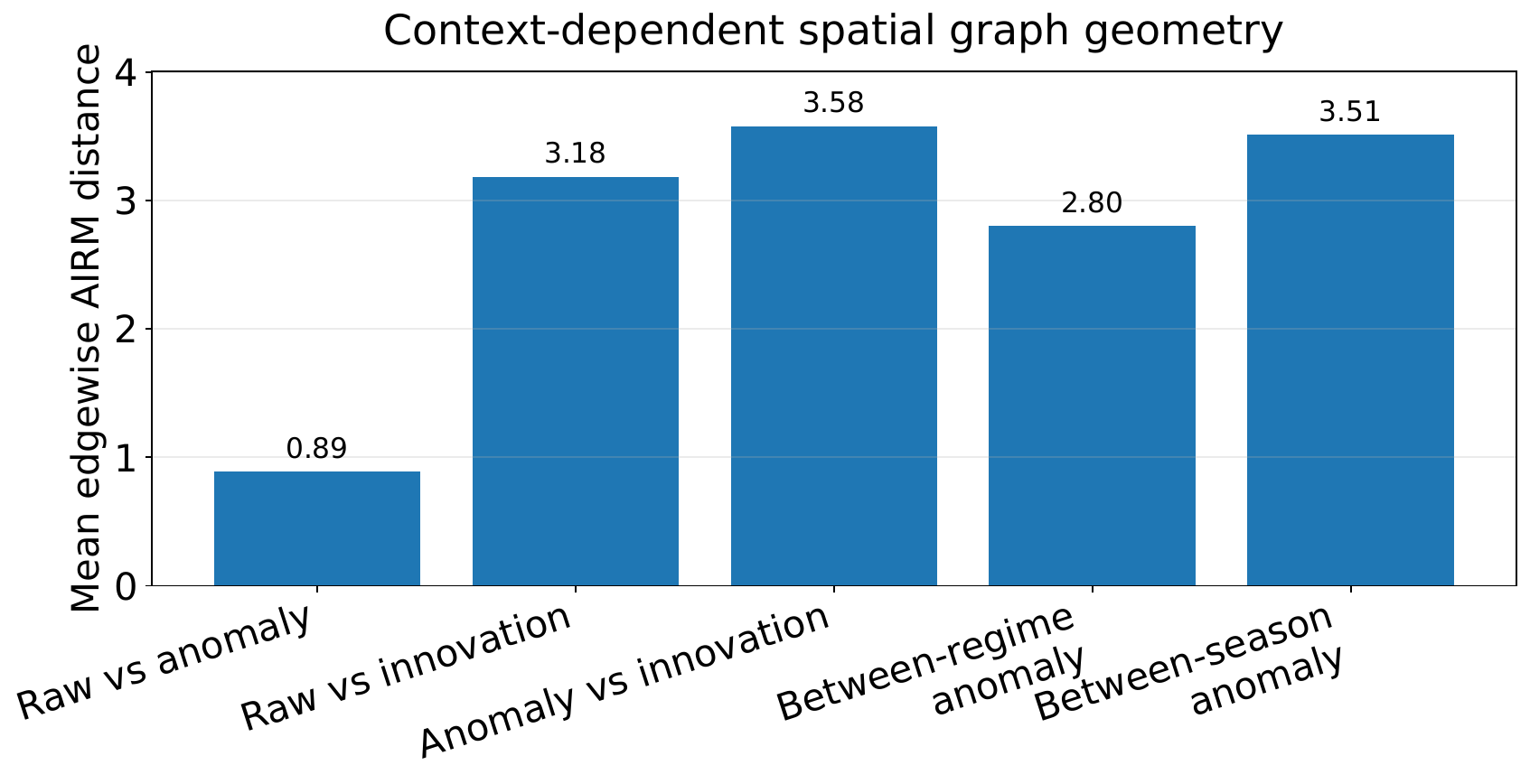}
\caption{Global and context separation.}
\end{subfigure}\hfill
\begin{subfigure}[t]{0.48\linewidth}\centering
\safegraphic{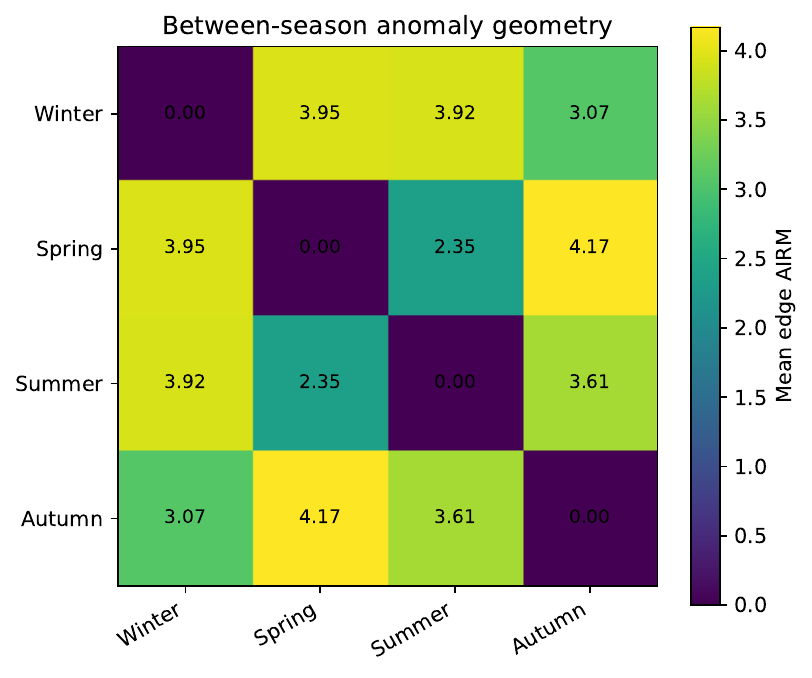}
\caption{Seasonal anomaly geometry.}
\end{subfigure}
\caption{
(A) Mean edgewise AIRM separation across the three global representations and
across training-derived regimes and seasons.
(B) Pairwise seasonal anomaly separation, largest between autumn and spring.
Distances are computed between corresponding posterior-mean edge matrices and
averaged over the 16 fixed geographic edges.
}
\label{fig:meteostat_context_geometry}
\end{figure}


A large AIRM distance by itself does not explain what has changed. \Cref{fig:meteostat_anisotropic_reconfiguration} therefore decomposes the global
anomaly-to-innovation difference into interpretable multivariate directions. For each edge $e$ and geometry
$g\in\{\mathrm{raw},\mathrm{anomaly},\mathrm{innovation}\}$, let
\[
\bar W_e^{(g)}
=
U_e^{(g)}
\operatorname{diag}
\!\left(
\lambda_{e1}^{(g)},\ldots,\lambda_{ed}^{(g)}
\right)
U_e^{(g)\top},
\qquad
\lambda_{e1}^{(g)}\ge\cdots\ge\lambda_{ed}^{(g)}.
\]
Panel~\ref{fig:meteostat_anisotropic_reconfiguration}(a) plots the leading-eigenvalue share
\[
s_e^{(g)}
=
\frac{\lambda_{e1}^{(g)}}
{\sum_{k=1}^{d}\lambda_{ek}^{(g)}}
=
\frac{\lambda_{e1}^{(g)}}
{\operatorname{tr}(\bar W_e^{(g)})},
\]
which measures how much of the total edge spectrum is concentrated in the
dominant ordinary eigenmode.  The variable associated with that mode is
identified separately from the corresponding eigenvector
$u_{e1}^{(g)}$: the mode is called pressure-dominated when the pressure
coordinate has the largest magnitude among the entries of
$u_{e1}^{(g)}$.  Thus the curve height measures dominance of the first
eigenmode, while the eigenvector coordinates determine which meteorological
variable dominates that mode.
 The leading ordinary eigenmode remains pressure-dominated on essentially every
edge: its mean spectral share is $0.846$, $0.848$, and $0.892$ for the raw,
anomaly, and innovation geometries, respectively. Hence the large
anomaly-to-innovation AIRM distance should not be interpreted as a complete
loss of common spatial structure. Instead, the data support a more specific
picture: a dominant pressure-related component remains stable, while other
multivariate directions reconfigure substantially.

Panel~\ref{fig:meteostat_anisotropic_reconfiguration}(b) resolves the
anomaly-to-innovation change edge by edge and direction by direction.  For
each edge $e$, the generalized eigenproblem
\[
W_e^{\rm innov}v_k
=
\lambda_k W_e^{\rm anomaly}v_k,
\qquad
\eta_k=\log\lambda_k,
\]
compares the two edge matrices along a common deformation direction $v_k$.
Thus $\eta_k>0$ indicates strengthening from anomaly to innovation along that
direction, whereas $\eta_k<0$ indicates weakening.  The dominant
meteorological variable is identified from the largest-magnitude coordinate
of the corresponding generalized eigenvector.  \Cref{fig:meteostat_anisotropic_reconfiguration}(b)
shows the strongest sign-certain modes, with horizontal intervals giving the
posterior 5-95\% quantiles.

For example, the Half Moon Bay-San Carlos edge has a strongly negative
wind-dominated mode,
\[
\eta=-5.557,
\qquad
90\%~\mathrm{posterior~interval}=[-7.330,-4.301],
\]
while the Palo Alto-San Carlos edge has a strongly negative
relative-humidity-dominated mode,
\[
\eta=-4.821,
\qquad
90\%~\mathrm{posterior~interval}=[-6.600,-3.460].
\]
These results show that the strongest anomaly-to-innovation changes are
localized to particular station pairs and meteorological directions rather
than arising from a uniform rescaling of all edge matrices.

\Cref{fig:meteostat_context_geometry,fig:meteostat_anisotropic_reconfiguration} therefore answer complementary
questions: the first measures how much the spatial graph changes across
representations and contexts, while the second identifies where those changes
occur and which variables dominate them.  This directional information is
specific to the matrix-valued representation and cannot be retained by a
single scalar edge weight.  Conditional-variance, effective-resistance, and
partial-correlation summaries give the same qualitative picture; full results
are reported in \Cref{app:meteostat}.

\begin{figure}[htbp]
\centering
\begin{subfigure}[t]{0.48\linewidth}\centering
\safegraphic{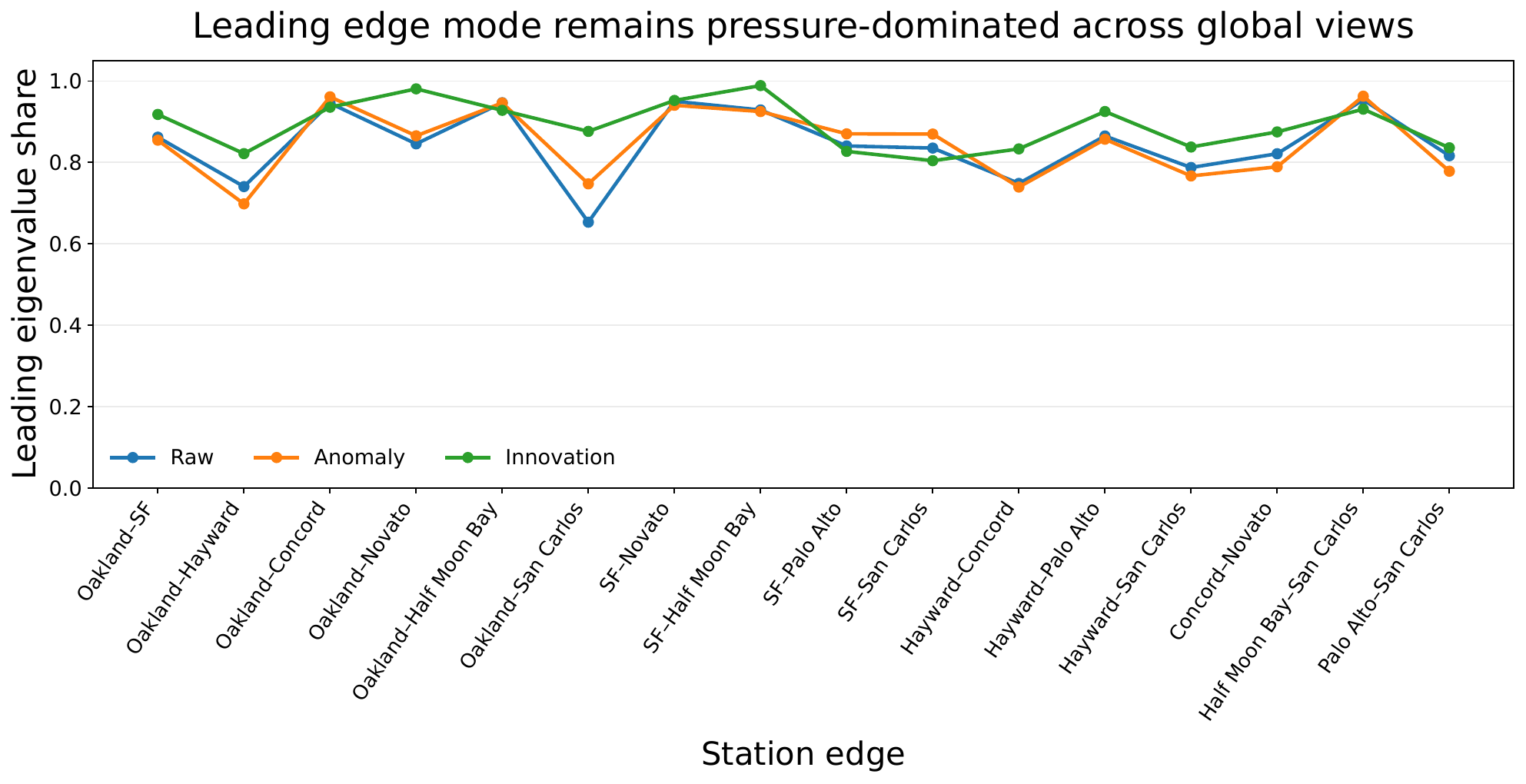}
\caption{Leading eigenmode share.}
\end{subfigure}\hfill
\begin{subfigure}[t]{0.48\linewidth}\centering
\safegraphic{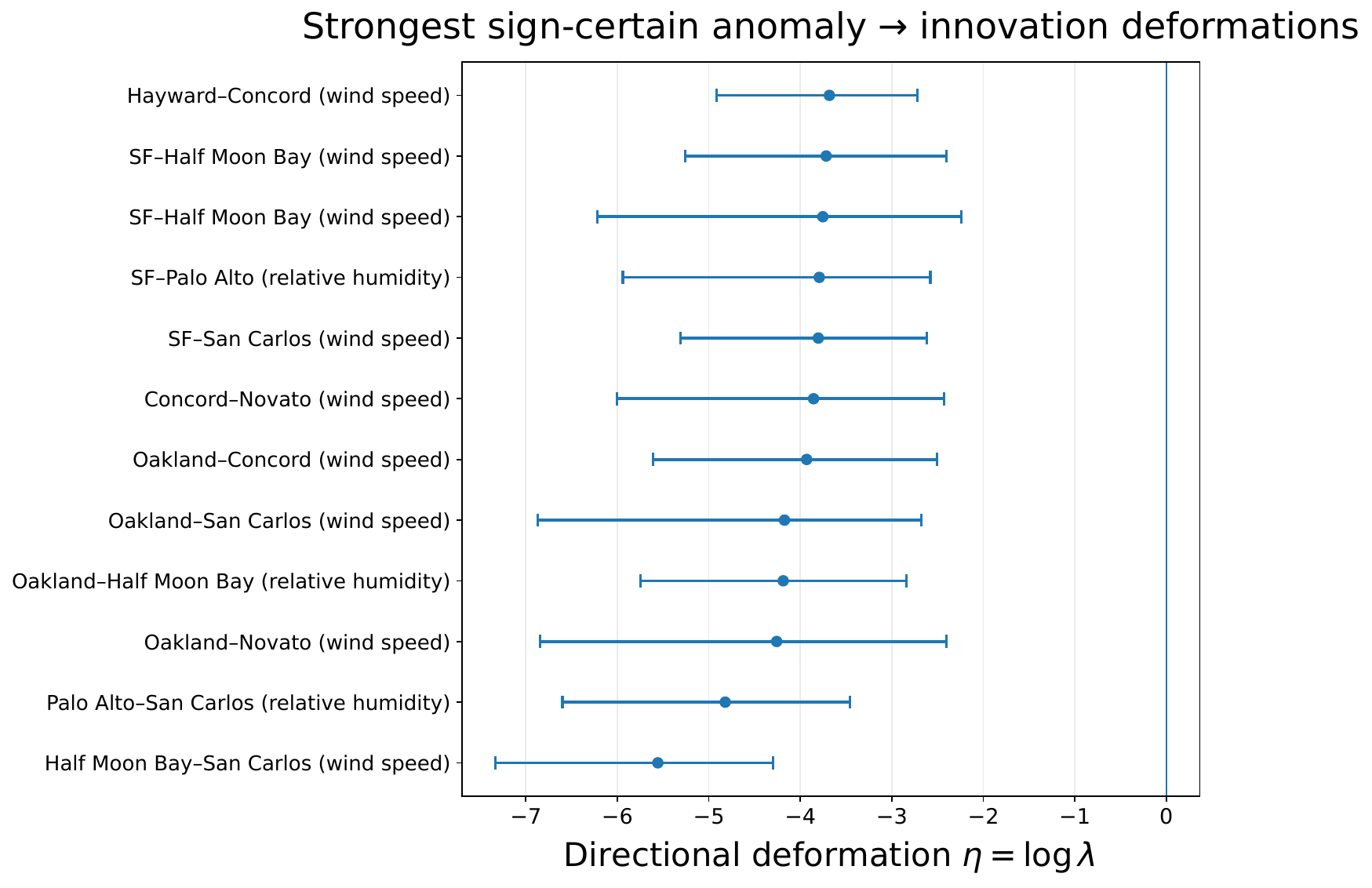}
\caption{Directional deformation.}
\end{subfigure}
\caption{
(A) The leading ordinary eigenmode accounts for a large fraction of each edge
matrix across raw, anomaly, and 12-hour innovation geometries; the corresponding
leading eigenvectors are pressure-dominated on essentially every edge.
(B) Strongest sign-certain anomaly-to-innovation generalized deformation modes.
Horizontal intervals are posterior 5-95\% quantiles; negative
$\eta=\log\lambda$ denotes weakening, positive $\eta$ strengthening, and the
labels identify the dominant meteorological variable.
}
\label{fig:meteostat_anisotropic_reconfiguration}
\end{figure}


These quantities should be interpreted as diagnostics of environmental
dependence rather than causal transport coefficients. They identify
where and along which measured-variable directions the learned
spatial organization changes, but they do not by themselves establish the
physical mechanism producing that change.

\subsection{Detecting atmospheric distribution shift}
\label{sec:ood}

As a secondary diagnostic, we ask whether the anomaly geometry learned from
the reference period can identify future structurally unusual days.
Out-of-distribution (OOD) labels are defined independently from feature-space
novelty. For a 3-hour field \(y\), the structural score is the
posterior-mean anomaly graph energy
\[
S_{\bar Q}(y)=y^\top\bar Q y,
\qquad
\bar Q=Q(\bar W),
\]
where \(\bar W\) is the posterior mean under the anomaly reference fit. We
evaluate the ranking OOD versus in-distribution (ID) days using the area
under the receiver operating characteristic curve (AUROC) and the area under
the precision-recall curve (AUPRC).  AUROC measures how well the score ranks
OOD days above ID days across all classification thresholds, with $0.5$
corresponding to random ranking and $1$ to perfect separation.  AUPRC summarizes
the tradeoff between precision and recall for the OOD class and is particularly
informative when the two classes are imbalanced.  Larger values of either
metric therefore indicate better discrimination of structurally unusual days.

Validation selects a trimmed-mean daily score that removes the largest
3-hour value before averaging.  On the chronologically held-out test period,
it achieves AUROC $0.835$ and AUPRC $0.848$. At the more conservative operating
point selected under a validation-ID false-positive-rate (FPR) budget of \(5\%\), the frozen threshold yields test true-positive rate (TPR) \(0.36\) and empirical FPR \(0\). Thus the score carries substantial ranking information, although only
\(36\%\) of OOD days are detected when false alarms are strongly constrained.
Because both the statistic and threshold are selected on validation data under
chronological splitting, we interpret these results as an empirical
operating-point analysis rather than a finite-sample error-control guarantee. Full distribution-shift detection results are reported in
\Cref{app:ood}.

\subsection{ER-associated gene-module reconfiguration}
\label{sec:tcga}

We analyze 1097 unique primary breast tumors from TCGA-BRCA through UCSC Xena
\citep{tcga,xena,xenaTCGABRCAHiSeqV2}. We use five prespecified
three-gene groups chosen as compact, literature-supported representatives of
major biological programs: Immune (CD8A, GZMB, STAT1)
\citep{ahn2017shc1,wong2022stat},
Metabolic (HK2, PKM, LDHA), and Structural/EMT (COL1A1, FN1, VIM),
the latter two supported by the corresponding MSigDB Hallmark gene sets
\citep{liberzon2015};
Signalling (EGFR, MAPK1, MYC)
\citep{sanchezvega2018,xu2010myc}; and
Epigenetic (DNMT1, EZH2, HDAC1)
\citep{oh2017epigenetic}.
These labels are used only to organize the 15-gene testbed; the three-gene
groups are not claimed to be complete biological pathways or learned gene
modules. The primary analysis uses common full-cohort centering for both ER
groups; an additional sensitivity analysis removes the ER-specific means while
keeping the same scaling and frozen edge maps.

The modules contain different genes, so directly subtracting, for example,
the first gene in one module from the first gene in another would impose an
arbitrary gene-to-gene matching. We avoid that. For tumor $k$, let
$y_i^{(k)}\in\mathbb R^3$ denote the standardized expression vector of the
three genes in module $i$. Before using any ER labels, we fit a regularized
pairwise canonical correlation analysis (CCA) for each of the ten module
pairs. For an edge $e=(i,j)$, CCA provides loading vectors
$a_{e,i},a_{e,j}\in\mathbb R^3$, chosen so that the projected module scores
have maximal regularized cross-module correlation. We retain only this leading
shared one-dimensional coordinate:
\[
s_{e,i}^{(k)}=a_{e,i}^{\top}y_i^{(k)},
\qquad
s_{e,j}^{(k)}=a_{e,j}^{\top}y_j^{(k)},
\]
and define the edge contrast
\[
\delta_e(y^{(k)})
=
s_{e,i}^{(k)}-s_{e,j}^{(k)}.
\]
The loading vectors are then frozen and reused in every ER fit, every topology,
and the mean-removal sensitivity analysis. Thus the ER labels cannot choose
the projection directions that make the two groups appear different.
The general fixed-map construction, including
arbitrary edge dimension $r_e$, is given in
\Cref{app:fixedmap}; The TCGA-BRCA analysis uses the conservative choice $r_e=1$. 

Each module pair therefore has one positive posterior weight $w_e$.  We use
the complete graph $K_5$ as the primary scaffold so that all
$\binom{5}{2}=10$ module pairs are present at the same time.  All genes are
centered and scaled once using the full primary-tumor cohort.  The clinical
field used here gives $n_{\rm ER-}=179,
\
n_{\rm ER+}=601,$ 
with 317 tumors lacking a usable ER label for this comparison.

We first ask whether the same fixed module-pair coordinates exhibit different
posterior coupling strengths in ER$-$ and ER$+$ tumors.  For each scalar edge,
we compare the posterior-mean weights through
\[
D_e
=
\left|
\log\frac{\bar w_e^{+}}{\bar w_e^{-}}
\right|,
\qquad
\eta_e
=
\log\frac{\bar w_e^{+}}{\bar w_e^{-}}.
\]
Here $D_e$ measures the magnitude of the ER-dependent reconfiguration, while
the sign of $\eta_e$ indicates whether the fitted coupling is stronger in
ER$+$ or ER$-$.  Posterior intervals for $\eta_e$ quantify whether the
direction of that change is supported despite posterior uncertainty.
 
We pair 4000 retained draws from the two fitted posteriors
to form the 5-95\% interval shown in
\Cref{fig:tcga-signed-eta}. 
The mean edge AIRM over $K_5$ is $0.537$, and six of ten 90\% signed-change
intervals exclude zero. Structural/EMT-Signalling is the strongest change,
with $\eta_e=1.625$ and interval $[1.173,2.314]$, corresponding to an
approximately $5.08$-fold posterior-mean weight ratio. Full edgewise values
are reported in \Cref{app:tcga}.
The remaining four edges have positive point estimates but intervals that
cross zero, and we therefore treat them as unresolved. All monitored chains in the reported analysis pass the numerical checks:
the largest split-$\Rhat$ is $1.0003$, the smallest ESS is $11{,}043$, and no
numerical failures occur.

\begin{figure}[htbp]
\centering
\includegraphics[height=0.20\textheight, keepaspectratio]{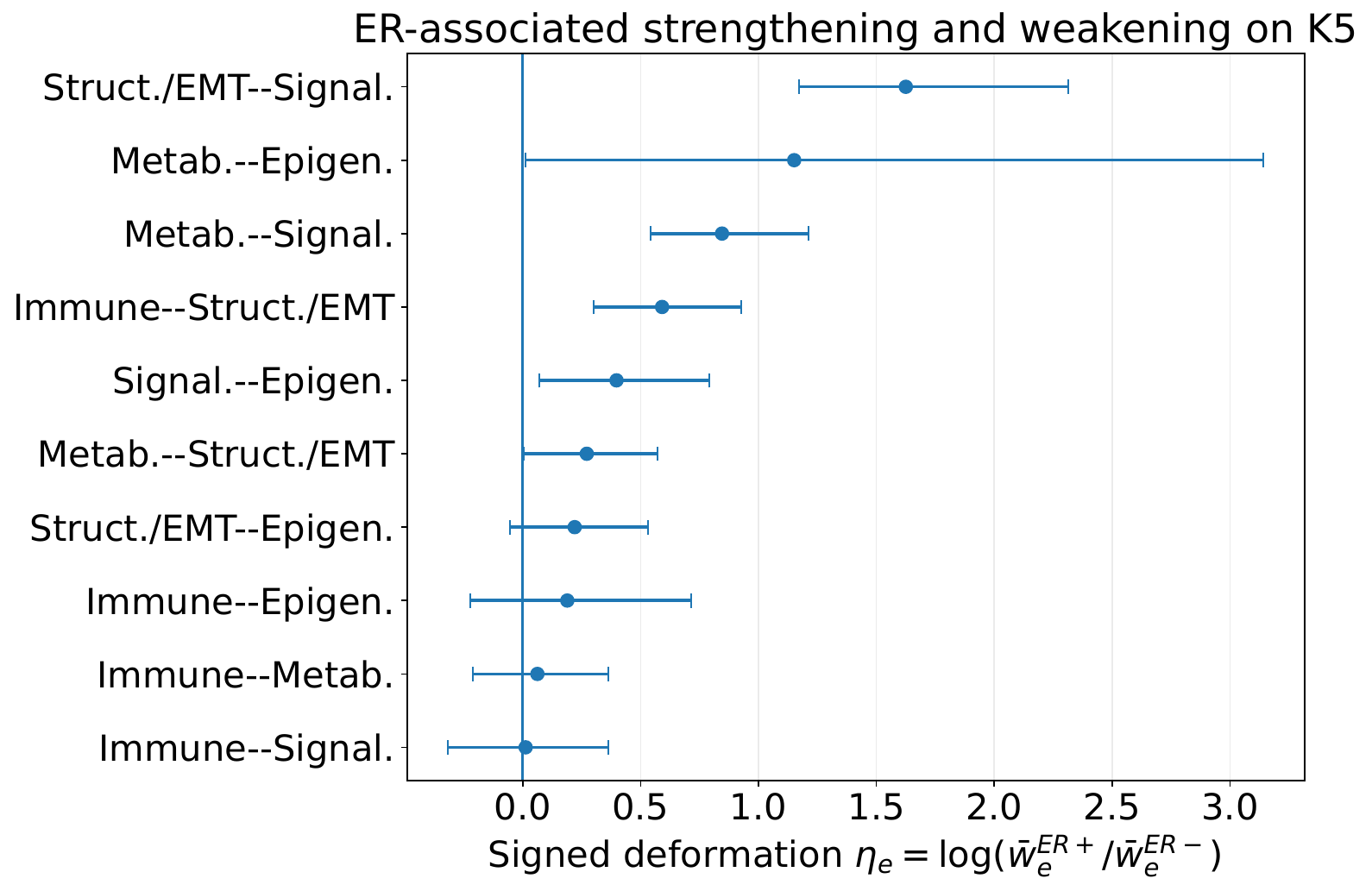}
\caption{
Each point is the signed change
$\eta_e=\log(\bar w_e^{+}/\bar w_e^{-})$ for one module pair; horizontal bars
are 90\% paired-posterior intervals. Values to the right of zero indicate
stronger fitted coupling in ER$+$ than ER$-$. Six of the ten intervals
exclude zero.}
\label{fig:tcga-signed-eta}
\end{figure}
Because BMVG conditions on a fixed graph scaffold, we test whether the
ER-associated edge reconfigurations depend strongly on the choice of that
scaffold. The primary $K_5$ model contains all ten possible module-pair
edges. As a topology-robustness check, we refit the model on every
five-cycle scaffold $C_5$, which uses the same five modules but retains only
five edges arranged in a closed loop. There are 12 distinct such cycles,
and each module-pair edge occurs in six of them.

For each edge $e$, let $D_e^{K_5}$ denote its AIRM distance in the primary
complete-graph fit and define $\overline D_e^{\,C_5}
=
\frac{1}{6}
\sum_{\substack{C_5\ni e}}
D_e^{C_5}$ 
as its average distance over the six five-cycle scaffolds containing $e$.
The two summaries agree closely across the ten module pairs
(Pearson $r=0.990$, Spearman $\rho=0.988$), and
Structural/EMT-Signalling ranks first in all six $C_5$ fits in which it
appears. Removing ER-specific means retains $94.1\%$ of the mean edge
separation and preserves all ten edge ranks (Spearman $\rho=1.000$).
Thus the leading ER-associated reconfigurations persist under both
substantial scaffold sparsification and removal of subgroup mean differences.
\begin{figure}[htbp]
\centering
\begin{subfigure}[t]{0.49\linewidth}
\centering
\safegraphic{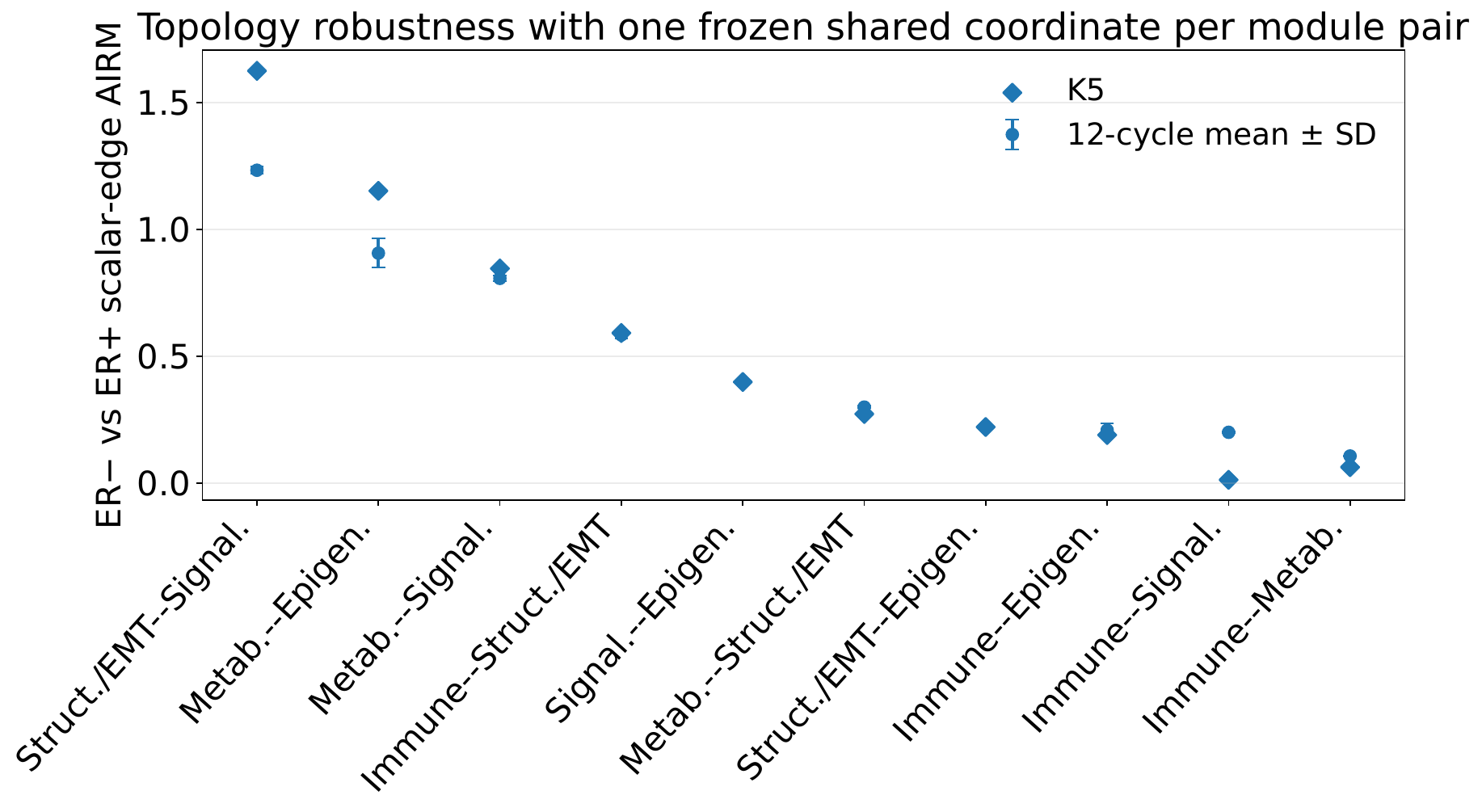}
\caption{Changing the graph scaffold.}
\end{subfigure}\hfill
\begin{subfigure}[t]{0.49\linewidth}
\centering
\safegraphic{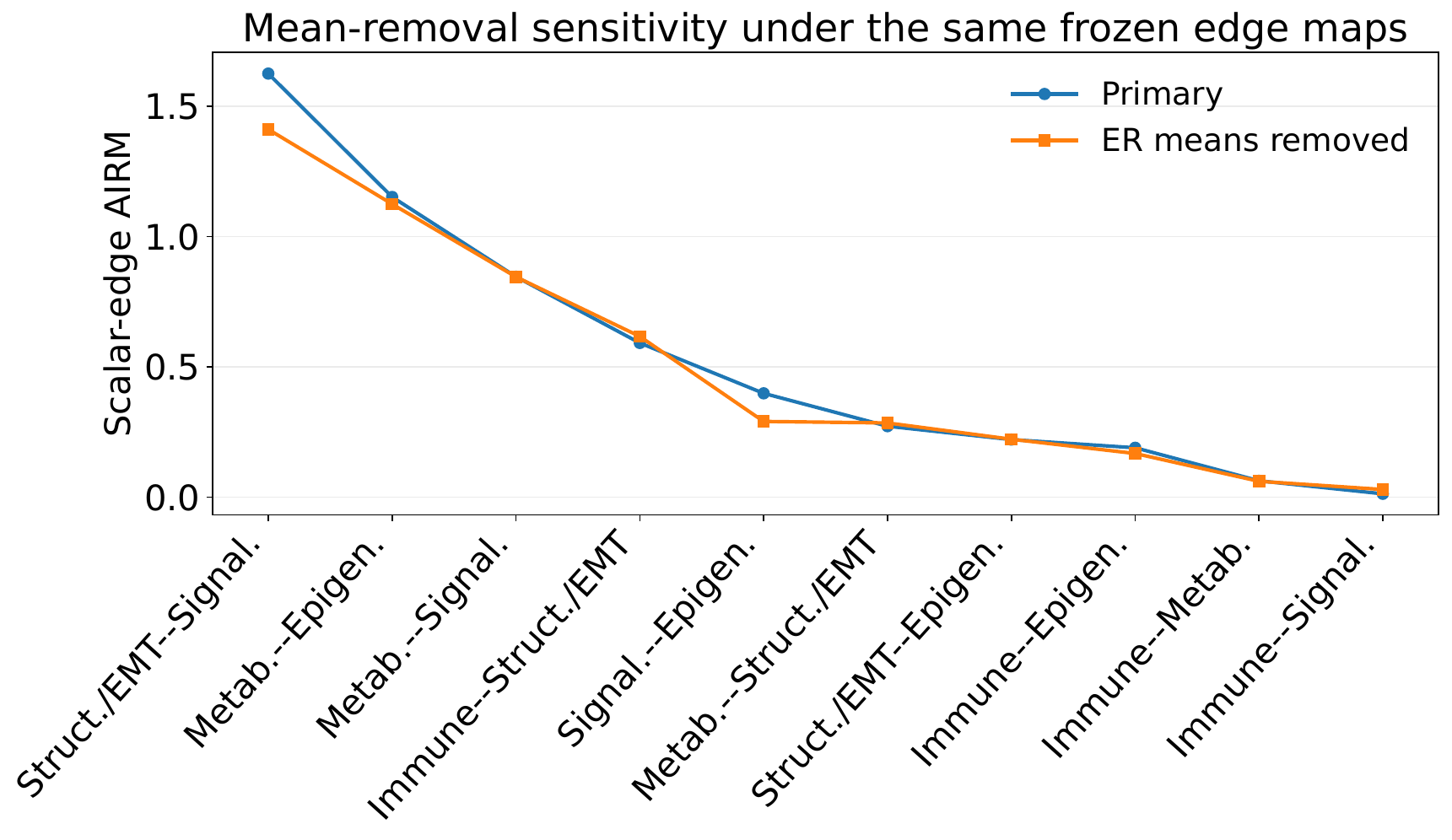}
\caption{Removing ER-specific means.}
\end{subfigure}
\caption{
(A) the complete-graph edge distances are compared with the mean and
standard deviation over the six five-cycles containing each edge.
(B) the primary and ER-mean-removed edge distances keep the same ordering.}
\label{fig:tcga-robustness}
\end{figure}

These results describe statistical dependence within this 15-gene system;
they do not establish causal interactions between modules, and the $K_5$
scaffold should not be interpreted as a biochemical interaction network.
Our conclusion is more specific: ER$-$ and ER$+$ tumors exhibit different
fitted dependence strengths for several module pairs, and the strongest
changes persist after altering the graph scaffold and removing ER-specific
mean shifts. Structural/EMT-Signalling shows the most stable reconfiguration
across these checks. Full edge tables, map audits, and sensitivity results
are reported in \Cref{app:tcga}.
\section{Discussion and limitations}
\label{sec:discussion}

The experiments show that matrix-valued edge inference captures
context-dependent changes that scalar edge weights cannot represent. In the
known-truth study, BMVG recovers orientation changes even when trace,
determinant, and ordinary eigenvalues are unchanged. In Meteostat, 12-hour
innovation geometry is substantially more separated from anomaly geometry
than raw state is from anomaly state, despite a stable pressure-dominated
leading mode. In TCGA-BRCA, the strongest ER-associated module-pair
reconfigurations persist under substantial scaffold sparsification and
removal of subgroup mean differences. Together, these results show that
posterior matrix-valued edge geometry can quantify not only how much a
relationship changes, but also which multivariate directions drive that
change.

The associated uncertainty summaries have distinct interpretations.
Posterior intervals quantify uncertainty under the fitted model, whereas in
the controlled study $R_e$ asks whether a planted change exceeds the
finite-sample separation observed between independently fitted no-change
posteriors. Neither quantity should be interpreted as a universal significance
criterion.

Several limitations remain. Graph topology is fixed rather than learned, so
all inference is conditional on the supplied scaffold. The Gaussian precision
likelihood is a working second-order model and does not imply causal
relationships. The controlled detection boundary depends on the chosen graph,
prior, sample sizes, and perturbation families. The TCGA analysis is restricted
to the specified 15-gene system and requires external-cohort validation before
broader biological conclusions are drawn. Finally, $\Rhat$ and ESS support
numerical stability of the sampled posteriors but do not provide a finite-time
MCMC convergence guarantee.
\paragraph{\textbf{Conclusion}.}
Matrix-valued graph posteriors characterize context-dependent multivariate
change through deformation magnitude, finite-sample uncertainty, and signed
directional reconfiguration. Across the controlled, Meteostat, and TCGA
studies, these quantities reveal changes that scalar edge summaries cannot
retain. Larger biological systems, external-cohort validation, and learned
graph scaffolds are important directions for future work.
\section*{Reproducibility Statement}
All reported experiments use fixed graph scaffolds, explicit priors,
prespecified diagnostic gates, and archived random seeds. \Cref{app:inference} reports the numerical sampler configuration underlying all reported results.
The controlled study records all generating seeds and nested sample sizes.
The Meteostat analysis records the station identifiers, chronological
splits, support transforms, training-derived climatology and regime model,
and graph construction.  The TCGA analysis records the source cohort,
gene list, normalization rule, CCA regularization, frozen maps,
topology-sensitivity fits, and mean-removal sensitivity settings.  Released result directories contain the fitted
configuration and posterior summaries used to generate each table and
figure. 
\paragraph{\textbf{Code availability}.}
Code and processed outputs are available at
\url{https://github.com/papridey/bayesian-matrix-valued-graphs}.

\section*{AI Use Statement}
Generative AI tools were used for language editing, organization, and
code-development assistance. All mathematical statements, experimental
designs, numerical results, code execution, verification, interpretation, and
final scientific claims remain the sole responsibility of the author.

\bibliographystyle{alpha}
\bibliography{references}

\appendix

\section{Theory and Derivations}
\label{app:theory}

\subsection{Sufficiency of the empirical second moment}
For independent centered observations
$Y_r\sim\mathcal N(0,Q^{-1})$,
\[
p(Y_1,\ldots,Y_n\mid Q)
\propto
(\det Q)^{n/2}
\exp\left\{
-\frac12
\sum_{r=1}^{n}
Y_r^\top QY_r
\right\}.
\]
Since
\[
\sum_{r=1}^{n}
Y_r^\top QY_r
=
\sum_{r=1}^{n}
\tr(QY_rY_r^\top)
=
n\tr(SQ),
\]
the likelihood can be written as
\[
p(Y_1,\ldots,Y_n\mid W)
\propto
\exp\left\{
\frac n2
[
\log\det Q(W)-\tr\{SQ(W)\}
]
\right\}.
\]
The observations enter the likelihood only through
$\sum_rY_rY_r^\top=nS$; the Fisher-Neyman factorization theorem therefore
gives sufficiency.

\subsection{All-edge score}
Using $Q(W)=R+\sum_{f\in E}B_fW_fB_f^\top,$ 
a perturbation $H_e$ of the edge block $W_e$ gives
\[
D_{W_e}Q(W)[H_e]
=
B_eH_eB_e^\top.
\]
Using the matrix differential identities
\[
d\log\det Q
=
\tr(Q^{-1}dQ),
\qquad
d\,\tr(SQ)
=
\tr(S\,dQ)
\]
and eq~\eqref{eq:likelihood}, we get
\[
D_{W_e}\ell(W)[H_e]
=
\frac n2
\tr\!\left[
\{Q^{-1}-S\}B_eH_eB_e^\top
\right].
\]
Using cyclic invariance of the trace,
\[
D_{W_e}\ell(W)[H_e]
=
\tr\!\left[
\frac n2
B_e^\top(Q^{-1}-S)B_e\,H_e
\right].
\]

By definition, the Euclidean matrix gradient with respect to $W_e$ is the
unique symmetric matrix $\nabla_{W_e}\ell(W)$ satisfying
\[
D_{W_e}\ell(W)[H_e]
=
\langle \nabla_{W_e}\ell(W),H_e\rangle_F
=
\tr\!\left[
\nabla_{W_e}\ell(W)^\top H_e
\right]
\]
for every symmetric perturbation $H_e$.  Since
$B_e^\top(Q^{-1}-S)B_e$ is symmetric, comparison of the two expressions
yields
\[
\nabla_{W_e}\ell(W)
=
\frac n2
B_e^\top(Q^{-1}-S)B_e,
\]
which is eq~\eqref{eq:edge-score}.

A perturbation $H_f$ of the edge block $W_f$ induces the perturbation
\[
H_Q
:=
D_{W_f}Q(W)[H_f]
=
B_fH_fB_f^\top
\]
of the full precision matrix.  Applying the chain rule to the composition
$W_f\mapsto Q(W)\mapsto Q(W)^{-1}$ and using
\[
D_Q(Q^{-1})[H_Q]
=
-Q^{-1}H_QQ^{-1},
\]
we obtain
\[
D_{W_f}[Q(W)^{-1}][H_f]
=
-Q^{-1}B_fH_fB_f^\top Q^{-1}.
\]
Note that $S$ does not depend on $W_f$. Therefore, we obtain
\[
D_{W_f}
\bigl[\nabla_{W_e}\ell(W)\bigr][H_f]
=
\frac n2 B_e^\top
D_{W_f}[Q(W)^{-1}][H_f]B_e
=
-\frac n2
B_e^\top Q^{-1}B_f
H_f
B_f^\top Q^{-1}B_e.
\]
For $e\neq f$ this term is generally nonzero, showing explicitly that the
likelihood couples distinct edge blocks through the common global precision
inverse $Q(W)^{-1}$.

\subsection{Gaussian Fisher geometry and the pullback metric}

For one observation with log likelihood
\[
\ell_Q(y)
=
\frac12\log\det Q
-
\frac12y^\top Qy+C,
\]
the directional score is
\[
D\ell_Q[H]
=
\frac12\tr(Q^{-1}H)
-
\frac12 y^\top Hy.
\]
Using Gaussian fourth-moment identities gives the Fisher information metric of the multivariate normal model~\cite{Sko84}, i.e., 
\[
g_Q^{\mathrm F}(H,K)
=
\E[
D\ell_Q[H]D\ell_Q[K]
]
=
\frac12
\tr(Q^{-1}HQ^{-1}K).
\]
For $n$ independent observations this is multiplied by $n$.

A tangent vector
$U=(U_e)_{e\in E}$ in edge space induces $L(U)=\sum_eB_eU_eB_e^\top.$
Substitution into the Fisher metric yields
\[
g_W^{\mathrm{pull}}(U,V)
=
\frac n2
\tr\left[
Q^{-1}L(U)Q^{-1}L(V)
\right],
\]
which contains cross-edge terms and is therefore not the block-diagonal
product metric used for the proposal.

\subsection{Reference measure of the product AIRM}

For one SPD matrix $W\in\SPD^d$, the affine-invariant Riemannian volume is,
up to a constant~\cite{Mui82,pennec2006},
\[
d\mathrm{vol}_{\AIRM}(W)
=
(\det W)^{-(d+1)/2}\,dW,
\]
where $dW$ is Lebesgue measure on the independent symmetric coordinates.
For all edges,
\[
d\mathrm{vol}_{\rm prod}(W)
\propto
\prod_{e\in E}
(\det W_e)^{-(d+1)/2}
\,dW_e.
\]

The Lebesgue Wishart factor contributes $(\det W_e)^{(\nu-d-1)/2}.$
Converting the posterior density from Lebesgue measure to product-AIRM volume
therefore multiplies it by
$(\det W_e)^{(d+1)/2}$, producing the power $\frac{\nu-d-1}{2}
+
\frac{d+1}{2}
=
\frac{\nu}{2}.$
This yields the $-\frac{\nu}{2}\log\det W_e$ term in eq~
\eqref{eq:intrinsic-potential}.

\subsection{Proof of Proposition~\ref{prop:distance-spectrum}}

Let $C=A^{-1/2}BA^{-1/2}\in\SPD^d.$ 
The generalized eigenvalues of $(B,A)$ are precisely the ordinary eigenvalues
$\lambda_k$ of $C$. Write $C=U\diag(\lambda_1,\ldots,\lambda_d)U^\top.$ 
Then $\log C
=
U
\diag(\log\lambda_1,\ldots,\log\lambda_d)
U^\top.$
Orthogonal invariance of the Frobenius norm gives
\[
d_{\AIRM}(A,B)^2
=
\|\log C\|_F^2
=
\sum_{k=1}^{d}
(\log\lambda_k)^2.
\]
This proves eq~\eqref{eq:distance-spectrum}.

\subsection{Fixed node-to-edge maps and unequal node dimensions}
\label{app:fixedmap}

The simple edge difference $y_i-y_j$ assumes that all nodes have the same
dimension and directly aligned coordinates, as in Meteostat.  For
heterogeneous nodes, let $y_i\in\R^{d_i}$ and
$p=\sum_{i\in V}d_i$.  For edge $e=(i,j)$, choose a shared edge-space
dimension $r_e$ and fixed linear maps
\[
A_{e,i}\in\R^{r_e\times d_i},
\qquad
A_{e,j}\in\R^{r_e\times d_j},
\]
and define
\[
\delta_e(y)=A_{e,i}y_i-A_{e,j}y_j.
\]

Let $D_e\in\R^{r_e\times p}$ be the block-row operator containing
$A_{e,i}$ in the columns of node $i$, $-A_{e,j}$ in those of node $j$,
and zeros elsewhere, so that $D_ey=\delta_e(y)$.  Define
$\widetilde B_e=D_e^\top$.  For $W_e\in\SPD^{r_e}$,
\begin{equation}
\mathcal E_e(y)
=
\delta_e(y)^\top W_e\delta_e(y)
=
y^\top\widetilde B_eW_e\widetilde B_e^\top y.
\label{eq:mapped-energy}
\end{equation}
Hence
\begin{equation}
L_A(W)
=
\sum_{e\in E}\widetilde B_eW_e\widetilde B_e^\top,
\qquad
Q_A(W)=L_A(W)+R.
\label{eq:mapped-precision}
\end{equation}
Since each summand in $L_A(W)$ is positive semidefinite and $R\succ0$,
$Q_A(W)\succ0$.

For the same Gaussian working likelihood
$Y_r\mid W\sim N(0,Q_A(W)^{-1})$, the edge gradient is
\begin{equation}
\nabla_{W_e}\ell(W)
=
\frac n2
\widetilde B_e^\top
\{Q_A(W)^{-1}-S\}
\widetilde B_e.
\label{eq:mapped-edge-score}
\end{equation}
Thus the shared precision solve, global cross-edge dependence, and
Metropolis-Hastings correction are unchanged; only $B_e$ is replaced by
$\widetilde B_e$.  With independent
$W_e\sim\mathcal W_{r_e}(\nu_e,\Psi_e)$, the intrinsic negative log
posterior relative to product-AIRM volume is
\begin{equation}
\Phi_A(W)
=
\frac n2\{\tr(SQ_A(W))-\log\det Q_A(W)\}
-\frac12\sum_{e\in E}\nu_e\log\det W_e
+\frac12\sum_{e\in E}\tr(\Psi_e^{-1}W_e).
\label{eq:mapped-intrinsic-potential}
\end{equation}

For contexts $a$ and $b$,
\[
W_e^bv_k=\lambda_kW_e^av_k,
\qquad
\eta_k=\log\lambda_k,
\]
and
\[
d_{\AIRM}(W_e^a,W_e^b)^2
=
\sum_{k=1}^{r_e}\eta_k^2.
\]

\paragraph{Example.}
If $y_1=(y_{11},y_{12})^\top$, $y_2=y_{21}$,
$A_{e,1}=[1\;2]$, and $A_{e,2}=[3]$, then
\[
\delta_e(y)=y_{11}+2y_{12}-3y_{21},
\qquad
\widetilde B_e=
\begin{bmatrix}
1&2&-3
\end{bmatrix}^{\!\top}.
\]
For a scalar edge weight $w_e>0$,
\[
\mathcal E_e(y)
=
w_e(y_{11}+2y_{12}-3y_{21})^2
=
y^\top\widetilde B_ew_e\widetilde B_e^\top y.
\]
This illustrates that the fixed-map construction first places heterogeneous
nodes in a common edge coordinate and then applies the same graph model.

\begin{proposition}
\label{prop:mapped-gauge}
Fix an edge $e$ and let $C_e\in\R^{r_e\times r_e}$ be invertible.  Replace
\[
A_{e,i}'=C_eA_{e,i},
\qquad
A_{e,j}'=C_eA_{e,j},
\qquad
W_e'=C_e^{-\top}W_eC_e^{-1}.
\]
Then the edge energy in \eqref{eq:mapped-energy}, the contribution of edge
$e$ to $Q_A(W)$, the AIRM distance between two context weights, and their
generalized eigenvalues are unchanged.
\end{proposition}

\begin{proof}
Since $\delta_e'(y)=C_e\delta_e(y)$,
\[
\delta_e'(y)^\top W_e'\delta_e'(y)
=
\delta_e(y)^\top W_e\delta_e(y).
\]
Equivalently,
$\widetilde B_e'=\widetilde B_eC_e^\top$ and therefore
$\widetilde B_e'W_e'\widetilde B_e'^\top
=\widetilde B_eW_e\widetilde B_e^\top$.
The remaining claims follow from congruence invariance of AIRM and of the generalized eigenvalues~\cite{bhatia2007,Moa05}.
\end{proof}

\begin{remark}
\label{rem:cca-normalization}
The maps $A_{e,i},A_{e,j}$ are fixed when the posterior over $W$ is
computed.  If maps and weights were learned jointly,
\Cref{prop:mapped-gauge} would yield multiple parameterizations with the
same edge energy, so additional normalization or orientation constraints
would be required for identifiability.
\end{remark}

In TCGA-BRCA, $d_i=3$ and $r_e=1$.  The maps are leading
regularized-CCA directions learned from pooled expression data without ER
labels and frozen before edge-weight inference.  Consequently $W_e=w_e>0$
and
\[
d_{\AIRM}(w_e^-,w_e^+)
=
\left|\log\frac{w_e^+}{w_e^-}\right|,
\qquad
\eta_e=\log\frac{w_e^+}{w_e^-}.
\]

\subsection{Identifiability of the edge blocks}
\label{sec:identifiability}

Throughout, $\mathcal{S}^{k}$ denotes the space of real symmetric $k\times k$
matrices and $\mathcal{S}^{k}_{++}$ the open cone of symmetric positive-definite
matrices, so $\dim\mathcal{S}^{k}=\binom{k+1}{2}$. Recall the graph precision
\[
  Q(W)=R+\sum_{e\in E}B_eW_eB_e^{\top},
  \qquad R\succ 0 \text{ fixed},\quad W_e\in\mathcal{S}^{r_e}_{++},
\]
and write $N$ for the ambient dimension: $N=md$ in the incidence model of
\Cref{sec:model}, and $N=p=\sum_{i\in V}d_i$ in the fixed-map model of
\Cref{app:fixedmap}, where $B_e$ is replaced by $\widetilde{B}_e=D_e^{\top}$.
Define the linear \emph{edge-assembly map}
\begin{equation}\label{eq:assembly-map}
  \mathcal{L}:\ \bigoplus_{e\in E}\mathcal{S}^{r_e}\longrightarrow \mathcal{S}^{N},
  \qquad
  \mathcal{L}(H)=\sum_{e\in E}B_eH_eB_e^{\top},\quad H=(H_e)_{e\in E}.
\end{equation}
Since $R$ is fixed, $Q(W)=R+\mathcal{L}(W)$, so all statements below transfer
verbatim to the fixed-map model upon replacing $B_e$ by $\widetilde{B}_e$.

\begin{lemma}
\label{lem:reduction}
Let $R\succ 0$ be fixed and known. Then $W=(W_e)_{e\in E}\in\prod_{e}\mathcal{S}^{r_e}_{++}$
is globally identifiable from the law of $Y\sim\mathcal{N}\!\bigl(0,Q(W)^{-1}\bigr)$
if and only if the map $\mathcal{L}$ in \eqref{eq:assembly-map} is injective on
$\bigoplus_{e}\mathcal{S}^{r_e}$.
\end{lemma}

\begin{proof}
For any $W\in\prod_e\mathcal{S}^{r_e}_{++}$ each summand $B_eW_eB_e^{\top}$ is
positive semidefinite and $R\succ0$, so $Q(W)\succ0$ and the law is a
nondegenerate centered Gaussian. Such a law is determined by, and determines,
its covariance $Q(W)^{-1}$, hence $Q(W)$. Because $W\mapsto\mathcal{L}(W)$ is
linear and $R$ is fixed,
\[
  \mathrm{law}(W)=\mathrm{law}(W')
  \iff Q(W)=Q(W')
  \iff \mathcal{L}(W-W')=0 .
\]
If $\mathcal{L}$ is injective, equal laws force $W=W'$, giving identifiability.
Conversely, suppose $\mathcal{L}(H)=0$ for some
$H=(H_e)\neq 0$ with $H_e\in\mathcal{S}^{r_e}$. Pick
$\tau>\tfrac12\max_{e}\lVert H_e\rVert_2$ and set
$W_e=\tau I_{r_e}+\tfrac12 H_e$ and $W_e'=\tau I_{r_e}-\tfrac12 H_e$. Then
$W_e,W_e'\succ 0$, $W\neq W'$, yet
$\mathcal{L}(W-W')=\mathcal{L}(H)=0$, so the two laws coincide and $W$ is not
identifiable. Global (not merely local) identifiability is automatic since
$\mathcal{L}$ is linear.
\end{proof}

\begin{theorem}
\label{thm:identif}
For each edge let
$\mathcal{V}_e:=\{\,B_eHB_e^{\top}:H\in\mathcal{S}^{r_e}\,\}\subseteq\mathcal{S}^{N}$.
Then $\mathcal{L}$ is injective, equivalently, $W$ is identifiable if and only
if both of the following hold:
\begin{enumerate}
  \item[\textup{(a)}] each $B_e$ has full column rank $r_e$; and
  \item[\textup{(b)}] the subspaces $\{\mathcal{V}_e\}_{e\in E}$ form a direct sum,
        i.e.\ $\mathcal{V}_e\cap\sum_{f\neq e}\mathcal{V}_f=\{0\}$ for every $e$;
        equivalently $\dim\sum_{e}\mathcal{V}_e=\sum_{e}\binom{r_e+1}{2}$.
\end{enumerate}
In particular, a necessary condition is the dimension count
\begin{equation}\label{eq:count}
  \sum_{e\in E}\binom{r_e+1}{2}\ \le\ \binom{N+1}{2}.
\end{equation}
\end{theorem}

\begin{proof}
\emph{Single edge.} The map $H\mapsto B_eHB_e^{\top}$ is injective on
$\mathcal{S}^{r_e}$ if and only if $B_e$ has full column rank. If $B_e^{+}$ is a
left inverse ($B_e^{+}B_e=I_{r_e}$) then
$B_eHB_e^{\top}=0\Rightarrow H=B_e^{+}(B_eHB_e^{\top})(B_e^{+})^{\top}=0$.
If instead $0\neq u\in\ker B_e$, then $H=uu^{\top}\neq0$ but
$B_eHB_e^{\top}=(B_eu)(B_eu)^{\top}=0$. When $B_e$ has full column rank,
$\dim\mathcal{V}_e=\dim\mathcal{S}^{r_e}=\binom{r_e+1}{2}$.

\emph{($\Leftarrow$).} Assume (a) and (b) and suppose
$\sum_{e}B_eH_eB_e^{\top}=0$. Set $X_e:=B_eH_eB_e^{\top}\in\mathcal{V}_e$; then
$\sum_e X_e=0$ with $X_e\in\mathcal{V}_e$, and the direct-sum property (b) forces
$X_e=0$ for every $e$. By (a) each $H_e=0$, so $\mathcal{L}$ is injective.

\emph{($\Rightarrow$).} Assume $\mathcal{L}$ injective. If some $B_e$ failed to
have full column rank, the single-edge argument would produce $H_e\neq0$ with
$B_eH_eB_e^{\top}=0$, hence $\mathcal{L}(0,\dots,H_e,\dots,0)=0$, contradicting
injectivity; thus (a) holds. Given (a), suppose (b) failed, so
$\sum_e X_e=0$ with $X_e\in\mathcal{V}_e$ not all zero. Writing
$X_e=B_eH_eB_e^{\top}$ and using (a) to recover $H_e$ (unique, and nonzero
exactly when $X_e\neq0$) yields $H=(H_e)\neq0$ with $\mathcal{L}(H)=0$, again a
contradiction; thus (b) holds.

Finally, (b) states $\dim\sum_e\mathcal{V}_e=\sum_e\binom{r_e+1}{2}$, and
$\sum_e\mathcal{V}_e\subseteq\mathcal{S}^{N}$ gives \eqref{eq:count}.
\end{proof}

We now specialize to the two models used in the paper. The first is
unconditional; the second has a genuine failure mode, certified away by a
finite rank check.

\begin{corollary}
\label{cor:incidence}
Let $B_e=b_e\otimes I_d$ with $b_e=\mathbf{e}_i-\mathbf{e}_j$ for
$e=\{i,j\}$, $G$ a simple graph, and any $d\ge 1$. Then $\mathcal{L}$ is
injective and $W$ is globally identifiable, with no condition on $G$ (beyond
simplicity) or on $d$.
\end{corollary}

\begin{proof}
By the Kronecker mixed-product rule \cite[Ch.~4]{HornJohnson1991},
\[
  (b_e\otimes I_d)\,H_e\,(b_e\otimes I_d)^{\top}
  =(b_eb_e^{\top})\otimes H_e .
\]
View $\mathcal{L}(H)=\sum_e (b_eb_e^{\top})\otimes H_e$ as an $m\times m$ array of
$d\times d$ blocks; its $(k,l)$ block equals $\sum_e (b_eb_e^{\top})_{kl}\,H_e$.
For $e=\{i,j\}$ one has $(b_eb_e^{\top})_{ij}=-1$, while
$(b_fb_f^{\top})_{ij}=0$ for every other edge $f$: in a simple graph the pair
$\{i,j\}$ carries at most one edge, and each $b_fb_f^{\top}$ has off-diagonal
support only on the endpoints of $f$. Hence $\mathcal{L}(H)=0$ forces the
$(i,j)$ block $-H_e=0$, i.e.\ $H_e=0$ for every edge $e$. Injectivity and
\Cref{lem:reduction} give the claim.
\end{proof}

This is the matrix-valued analogue of the classical fact that the off-diagonal
entries of a weighted graph Laplacian are the negated edge weights, so the
weights are read off uniquely \cite[Ch.~13]{GodsilRoyle2001}; cf.\ the Gaussian
graphical parameterization in \cite{lauritzen1996}. Consequently, in the incidence model the edge-block parameterization is
globally identifiable for any simple graph and any $d$. This structural
identifiability is independent of the finite-sample recovery accuracy observed
in the known-truth study.

\begin{corollary}
\label{cor:scalar}
In the fixed-map model with $r_e=1$, write $B_e=\widetilde{B}_e=v_e\in\mathbb{R}^{p}$,
where $v_e$ places $a_{e,i}$ in the block of node $i$ and $-a_{e,j}$ in the block
of node $j$, and $W_e=w_e>0$, so that
$\mathcal{L}\bigl((w_e)\bigr)=\sum_{e}w_e\,v_ev_e^{\top}$. Then $W$ is identifiable
if and only if the rank-one matrices $\{v_ev_e^{\top}\}_{e\in E}$ are linearly
independent in $\mathcal{S}^{p}$. Equivalently:
\begin{enumerate}
  \item[\textup{(i)}] \emph{(necessary)} $|E|\le\binom{p+1}{2}$;
  \item[\textup{(ii)}] \emph{(checkable)} let $M\in\mathbb{R}^{\binom{p+1}{2}\times|E|}$
        have columns $\operatorname{svec}(v_ev_e^{\top})$, where
        $\operatorname{svec}:\mathcal{S}^{p}\to\mathbb{R}^{\binom{p+1}{2}}$ is the
        symmetric vectorization; then $W$ is identifiable iff
        $\operatorname{rank}(M)=|E|$;
        \item[\textup{(iii)}] \emph{(generic sufficiency)}
whenever $|E|\le\binom{p+1}{2}$, linear independence holds for
Lebesgue-almost-every $(v_e)_{e\in E}$; equivalently, full rank is generic
when the edge-map vectors are in general position.
\end{enumerate}
\end{corollary}

\begin{proof}
Apply \Cref{thm:identif} with $r_e=1$. Condition (a) is $v_e\neq0$, i.e.\
full column rank of a nonzero vector. Condition (b) is that the lines
$\mathbb{R}\,v_ev_e^{\top}$ form a direct sum, i.e.\ that
$\{v_ev_e^{\top}\}$ are linearly independent. Since
$\operatorname{svec}$ is a linear isomorphism $\mathcal{S}^{p}\cong\mathbb{R}^{\binom{p+1}{2}}$,
independence is exactly $\operatorname{rank}(M)=|E|$, giving (ii); the necessary
count \eqref{eq:count} reads $|E|\le\binom{p+1}{2}$, giving (i).

For (iii), suppose $|E|\le\binom{p+1}{2}$. Independence fails only when all
$|E|\times|E|$ minors of $M(v)$ vanish; these minors are polynomials in the
entries of $(v_e)$, so the dependence set is an algebraic subvariety of the
$(v_e)$-space. It is a \emph{proper} subvariety provided independence holds for
at least one choice of $(v_e)$. This is the case: the outer products
$\{vv^{\top}:v\in\mathbb{R}^{p}\}$ span $\mathcal{S}^{p}$ (e.g.\
$\mathbf{e}_i\mathbf{e}_i^{\top}$ and
$(\mathbf{e}_i+\mathbf{e}_j)(\mathbf{e}_i+\mathbf{e}_j)^{\top}$ together yield all
$E_{ii}$ and $E_{ij}+E_{ji}$), so one may select $|E|\le\binom{p+1}{2}$ vectors
with linearly independent outer products. A proper algebraic subvariety has Lebesgue measure zero, establishing genericity~\cite{DSS09}.
\end{proof}
\begin{remark}\label{rem:tcga}
In the TCGA-BRCA analysis, $p=15$ and
$|E|=10\leq \binom{16}{2}=120$. For the specific frozen
regularized-CCA maps used in the analysis, the matrix $M$ in
\Cref{cor:scalar} has numerical rank $10$. Hence the ten scalar
edge weights $(w_e)$ are identifiable for these fixed maps.
The smallest singular value of $M$ is $1.335$, with condition number
approximately $3.10$, so the full-rank conclusion is numerically well
separated from the rank-deficient case. \Cref{cor:scalar}(iii)
further shows that full rank is generic when the edge-map vectors are in
general position.
\end{remark}
\subsection{Why $\Prob(D_e>0)$ is not an equality test}

Suppose $\Pi_a$ and $\Pi_b$ have densities on $\SPD^d$.  For independent
$W_a\sim\Pi_a$ and $W_b\sim\Pi_b$, the event $W_a=W_b$ has probability
zero.  Since AIRM is a metric,
\[
d_{\AIRM}(W_a,W_b)=0
\quad\Longleftrightarrow\quad
W_a=W_b,
\]
and hence
\[
\Prob\{d_{\AIRM}(W_a,W_b)>0\}=1,
\]
even when the underlying data-generating parameters are identical.  A
positive draw-pair distance therefore mixes structural separation with
finite-sample posterior spread, motivating the same-truth reference used
below.

\section{Posterior Inference Algorithm}
\label{app:inference}

For the intrinsic potential \eqref{eq:intrinsic-potential}, define
\begin{equation}
G_e(W)
=
\frac n2B_e^\top(S-Q(W)^{-1})B_e
-\frac{\nu}{2}W_e^{-1}
+\frac12\Psi^{-1}.
\label{eq:potential-gradient}
\end{equation}
Under the edgewise AIRM,
\[
\operatorname{grad}_e\Phi(W)=W_eG_e(W)W_e.
\]
If
$\mathcal B=[B_{e_1}\;\cdots\;B_{e_{|E|}}]$, a single solve
$QZ=\mathcal B$ after Cholesky factorization of $Q$ provides all
$B_e^\top Q^{-1}B_e$.

\begin{algorithm}[htbp]
\caption{Posterior inference for one graph context}
\label{alg:posterior-context}
\begin{algorithmic}[1]
\Require Centered observations $\{Y_r\}_{r=1}^n$; graph $G$; stabilizer $R$;
prior $(\nu,\Psi)$; sampler settings.
\Ensure Posterior draws and multi-chain diagnostics.
\State Form
\[
S=n^{-1}\sum_rY_rY_r^\top,
\qquad
Q(W)=R+\sum_eB_eW_eB_e^\top.
\]
\For{each chain}
    \State Initialize the SPD edge blocks.
    \For{each MCMC iteration}
        \State Evaluate the intrinsic posterior and product-AIRM gradient
        using the shared precision solve.
        \State Generate a product-AIRM Langevin proposal and map each block
        with the affine-invariant exponential map.
        \State Apply the exact Metropolis-Hastings correction.
    \EndFor
    \State Discard burn-in and retain the prescribed production draws.
\EndFor
\State Compute split-$\widehat R$, ESS, acceptance, and numerical-failure
summaries.
\end{algorithmic}
\end{algorithm}

The affine-invariant exponential map is
\[
\operatorname{Exp}_W(V)
=
W^{1/2}
\exp(W^{-1/2}VW^{-1/2})
W^{1/2},
\]
which remains SPD for every symmetric tangent vector $V$.  Full
proposal-density and reversibility derivations are given in the dedicated
sampler paper \cite{sampler}.

Independent Wishart priors $W_e\sim\mathcal W_d(\nu,\Psi)$ are used
throughout.  For $d=3$, $\nu=5$ and $\Psi=I_3/5$; for Meteostat
($d=4$), $\nu=6$ and $\Psi=I_4/6$, giving prior mean $I_d$.

\section{Additional Results for Known-Truth Structural Resolution}
\label{app:expA}

We use five independent generating seeds,
$n\in\{200,500,1000\}$, and target distances
$D^\star\in\{0.5,0.75,1.0,1.5\}$.  For each seed, a fresh baseline truth
is generated and held fixed while the controlled perturbations are applied.
An independent same-truth dataset provides the finite-sample no-change
reference.  Smaller sample sizes use nested prefixes of the corresponding
$n=1000$ dataset.

\begin{algorithm}[htbp]
\caption{Controlled structural-change experiment with a no-change reference}
\label{alg:controlled-change}
\begin{algorithmic}[1]

\Require Independent generating seeds $j=1,\ldots,J$; sample-size grid
$\mathcal N$; target-distance grid $\mathcal G_D$; perturbation families;
evidence cutoff $\gamma=0.95$.

\Ensure Paired AIRM-distance summaries $D_{e,\mathrm{ref}}$ and $D_{e,c}$,
posterior comparison probabilities $R_e(c)$, detection indicators, and
matched generalized deformation modes.

\For{each generating replicate $j$}
    \State Generate a fresh baseline truth $W_j^{\star,0}$.
    \State Set $W_j^{\star,0'}=W_j^{\star,0}$ to define a same-truth
    reference condition.

    \For{each $D^\star\in\mathcal G_D$}
        \State Construct a sparse perturbation of one edge with target AIRM
        distance $D^\star$.
        \State Construct an isospectral orientation perturbation
        $W_e^{\star,c}=RW_e^{\star,0}R^\top$, $R^\top R=I$, with the same
        target distance.
    \EndFor

    \State Generate independent $n_{\max}=\max\mathcal N$ observations from
    the baseline, same-truth reference, and each perturbed condition; use
    nested prefixes for smaller $n$.

    \For{each $n\in\mathcal N$}
        \State Fit all required posteriors using
        \Cref{alg:posterior-context} and apply the diagnostic gate.

        \For{each edge $e$}
            \State Pair baseline and same-truth posterior draws to obtain
            $D_{e,\mathrm{ref}}$.

            \For{each perturbed condition $c$}
                \State Pair baseline and condition draws to obtain $D_{e,c}$.
                \State Estimate
                \[
                R_e(c)
                =
                \Prob(D_{e,c}>D_{e,\mathrm{ref}}\mid\mathcal D)
                \]
                by Monte Carlo pairing.
                \State Declare the change detectable when $R_e(c)\ge\gamma$.
                \State Record deformation magnitude and generalized modes,
                and match identifiable directions to the planted directions.
            \EndFor
        \EndFor
    \EndFor
\EndFor

\State Aggregate sensitivity, specificity, no-change separation, magnitude,
sign, and direction recovery.
\end{algorithmic}
\end{algorithm}

\subsection{Detection boundary}

The rule $R_e\ge0.95$ declares a planted change detectable only when its
draw-pair distance exceeds the finite-sample separation observed under the
same underlying graph.  Each condition changes one of six edges, so five
generating seeds provide five changed-edge and $25$ unchanged-edge
comparisons.  No unchanged edge is declared changed in any reported
condition, giving specificity $1.000$ throughout.

\begin{table}[H]
\centering
\small
\setlength{\tabcolsep}{5pt}
\begin{tabular}{llcccc}
\toprule
$n$ & Family &
$D^\star=0.50$ &
$D^\star=0.75$ &
$D^\star=1.00$ &
$D^\star=1.50$ \\
\midrule
200  & Sparse      & $0/5$ & $0/5$ & $0/5$ & $1/5$ \\
200  & Orientation & $0/5$ & $0/5$ & $0/5$ & $4/5$ \\
500  & Sparse      & $0/5$ & $1/5$ & $1/5$ & $5/5$ \\
500  & Orientation & $0/5$ & $2/5$ & $4/5$ & $5/5$ \\
1000 & Sparse      & $0/5$ & $1/5$ & $3/5$ & $5/5$ \\
1000 & Orientation & $0/5$ & $4/5$ & $5/5$ & $5/5$ \\
\bottomrule
\end{tabular}
\caption{
Fraction of five independent replicates in which the planted edge satisfies
$R_e\ge0.95$.
}
\label{tab:expA-boundary}
\end{table}

\subsection{No-change separation and deformation recovery}

Even without structural change, independent finite datasets produce positive
posterior separation.  Averaged over the six edges and five generating
replicates, the mean AIRM distance decreases from $1.052$ at $n=200$ to
$0.449$ at $n=1000$.

\begin{table}[htbp]
\centering
\small
\begin{tabular}{ccc}
\toprule
$n$ & Mean AIRM distance & Mean 95th percentile \\
\midrule
200  & $1.052$ & $1.726$ \\
500  & $0.636$ & $0.991$ \\
1000 & $0.449$ & $0.693$ \\
\bottomrule
\end{tabular}
\caption{
Finite-sample separation between posteriors fitted to independent datasets
generated from the same graph, averaged over six edges and five replicates.
}
\label{tab:expA-null}
\end{table}

At $D^\star=1$, magnitude recovery moves toward the planted value as $n$
increases, while sign recovery is already close to one and direction
alignment improves systematically.

\begin{table}[htbp]
\centering
\footnotesize

\begin{subtable}[t]{0.38\linewidth}
\centering
\begin{tabular}{ccc}
\toprule
$n$ & Sparse & Orientation \\
\midrule
200  & $1.561$ & $1.446$ \\
500  & $1.320$ & $1.363$ \\
1000 & $1.065$ & $1.138$ \\
\bottomrule
\end{tabular}
\caption{Recovered AIRM magnitude.}
\end{subtable}
\hfill
\begin{subtable}[t]{0.58\linewidth}
\centering
\begin{tabular}{lccc}
\toprule
Metric & $n=200$ & $n=500$ & $n=1000$ \\
\midrule
Sparse sign              & $0.996$ & $1.000$ & $1.000$ \\
Sparse alignment         & $0.853$ & $0.928$ & $0.956$ \\
Orientation sign         & $0.991$ & $0.999$ & $1.000$ \\
Orientation alignment    & $0.928$ & $0.963$ & $0.978$ \\
\bottomrule
\end{tabular}
\caption{Directional recovery.}
\end{subtable}

\caption{
Recovery at $D^\star=1$.  Sign entries give the probability of recovering
the planted strengthening or weakening direction; alignment is the
sign-invariant baseline-metric alignment with the planted generalized
direction.
}
\label{tab:expA-recovery}
\end{table}

All $165$ posterior fits pass the prespecified diagnostic gate: production
acceptance lies in $[0.644,0.746]$, the largest split-$\widehat R$ is
$1.0012$, the smallest ESS is $6448$, and no numerical failures occur.


\subsection{Native-comparator details and full-grid behavior}
\label{app:native_comparison_details}
This appendix expands the native-method comparison in
\Cref{sec:expA_native_comparison}. For each condition, BMVG estimator uses posterior samples from the corresponding matrix-valued graph fit, obtained with ConeMALA. DWW14 FGL is fit to the two
\(15\)-dimensional contexts using the native \texttt{JGL} fused graphical-lasso
solver; its sparsity and fusion penalties are selected by validation Gaussian
NLL on a grid that includes the unpenalized boundary. PSV15 multiGGM is fit
with the G-Wishart sampler using four chains, \(5000\) burn-in iterations and
\(2000\) saved draws per chain in the reported comparison. The fit uses only the
zero-mean training scatter matrices \(X_c^\top X_c\) and sample sizes; planted
truth and test observations are withheld from fitting. Flury84 CPC is fit by
profiled Gaussian maximum likelihood to the three-dimensional changed-edge
contrast and is evaluated on contrast covariance, covariance change, and
leading-axis rotation. These choices intentionally retain each method's native
statistical target rather than forcing all competitors into BMVG 
matrix-edge parameterization.

\Cref{fig:native_deltaQ} shows relative \(\Delta Q\) error at \(n=1000\) across
the full deformation grid. Under sparse perturbations, BMVG 
improves rapidly with signal strength and is more accurate than both FGL and
PSV15 for moderate and large \(D^\star\). Under orientation perturbations,
FGL is competitive and can be more accurate for the global
difference matrix, especially at smaller deformation magnitudes. This does
not contradict the directional result in the main text: a fused estimate can
approximate the aggregate matrix difference without recovering the
matrix-valued edge deformation or its rotating directions.

\begin{figure*}[htbp]
\centering
\begin{minipage}[t]{0.49\linewidth}
\centering
\includegraphics[width=\linewidth]{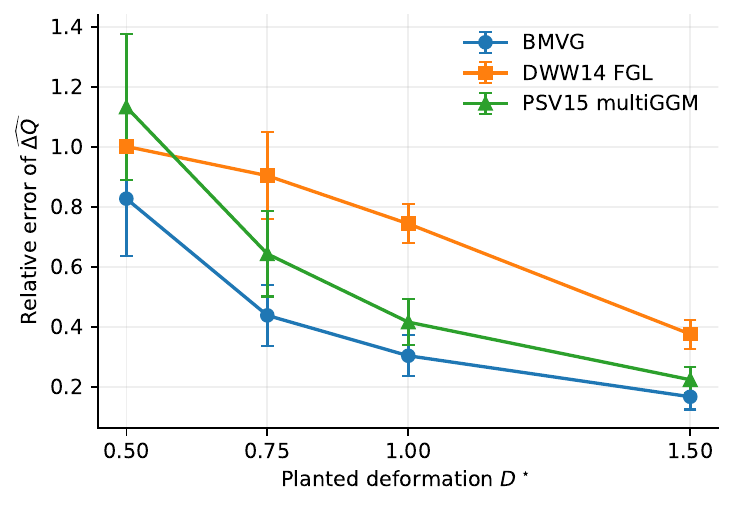}\\[-2pt]
\textbf{(a)} Sparse perturbation
\end{minipage}\hfill
\begin{minipage}[t]{0.49\linewidth}
\centering
\includegraphics[width=\linewidth]{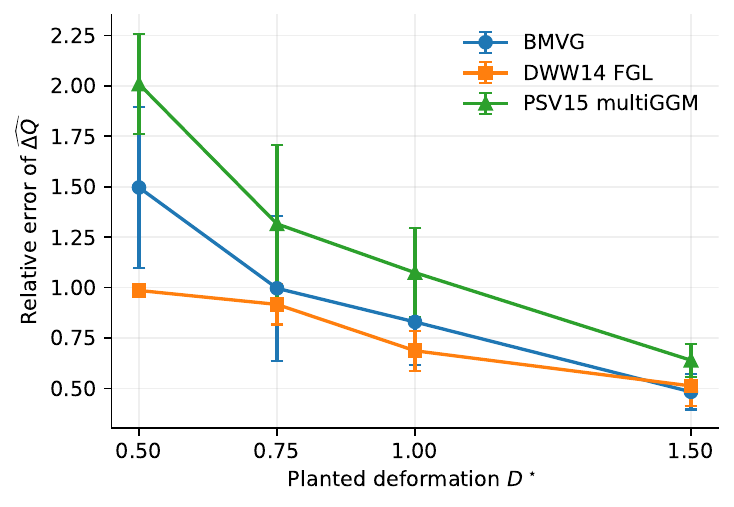}\\[-2pt]
\textbf{(b)} Orientation perturbation
\end{minipage}
\caption{Global differential-precision recovery at $n=1000$.
The vertical axis is
$\|\widehat{\Delta Q}-\Delta Q^\star\|_F/\|\Delta Q^\star\|_F$.
(a) Under sparse perturbations, BMVG has the smallest error once the
deformation is moderate. (b) Under orientation perturbations, FGL is
competitive for the aggregate precision change and is more accurate over
part of the deformation range. These global $\Delta Q$ results complement
the edge-level directional comparison in the main text. Error bars show
one sample standard deviation across five independent replicates.}
\label{fig:native_deltaQ}
\end{figure*}

For localization we score each module pair by the Frobenius norm of its
posterior-mean precision change. BMVG ranks only the six edges
in its fixed scaffold, whereas the native FGL and PSV15 comparisons search all
ten module pairs. To separate this candidate-set advantage from the fitted
precision estimates, we also report an auxiliary FGL score obtained by
restricting the already-fitted FGL result to the same six scaffold edges.
Across the complete \(120\)-condition grid, BMVG has mean rank
\(1.258/6\), AUROC \(0.948\), and AUPRC \(0.914\). PSV15 is also a strong
localizer when its posterior \(\Delta Q\) is aggregated into \(3\times3\)
module blocks, with mean rank \(1.433/10\), AUROC \(0.952\), and AUPRC
\(0.882\). FGL is weaker at small and moderate sample sizes but converges to
near-perfect localization for sufficiently strong signals.

\begin{table*}[htbp]
\centering
\scriptsize
\setlength{\tabcolsep}{4.5pt}
\begin{tabular}{llccc}
\toprule
Method & Localization score / candidate set & Mean rank & AUROC & AUPRC\\
\midrule
BMVG & 6-edge scaffold & 1.258 & 0.948 & 0.914 \\
DWW14 FGL & all 10 pairs & 3.625 & 0.708 & 0.480 \\
DWW14 FGL & post-fit 6-edge restriction & 2.454 & 0.709 & 0.525 \\
PSV15 multiGGM & all 10 pairs; \(\|\Delta Q_{ij}\|_F\) &
1.433 & 0.952 & 0.882 \\
PSV15 multiGGM & all 10 pairs; mean differential PIP &
7.325 & 0.297 & 0.143 \\
\bottomrule
\end{tabular}
\caption{Values average over all seeds, sample sizes, effect sizes, and both
perturbation families. Candidate sets differ by design: BMVG uses the fixed
six-edge scaffold, whereas native DWW14 and PSV15 search all ten module
pairs. The post-fit FGL restriction is an auxiliary same-scaffold
diagnostic.}
\label{tab:native_localization}
\end{table*}

\begin{figure*}[htbp]
\centering
\begin{minipage}[t]{0.49\linewidth}
\centering
\includegraphics[width=\linewidth]{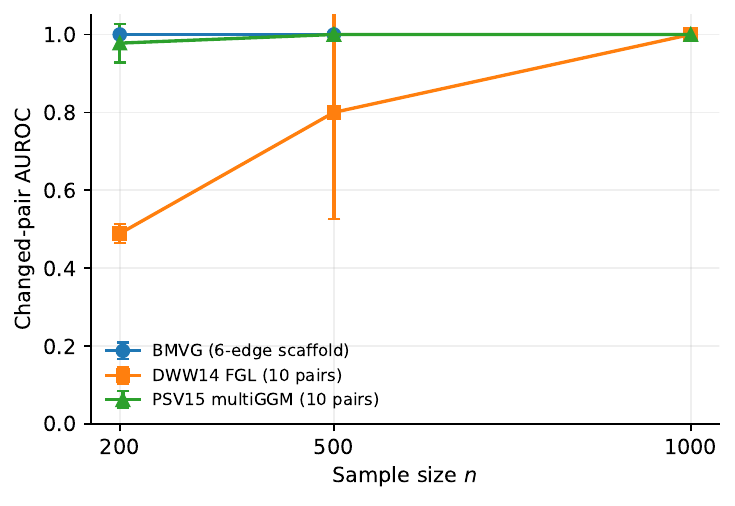}\\[-2pt]
\textbf{(a)} Sparse perturbation
\end{minipage}\hfill
\begin{minipage}[t]{0.49\linewidth}
\centering
\includegraphics[width=\linewidth]{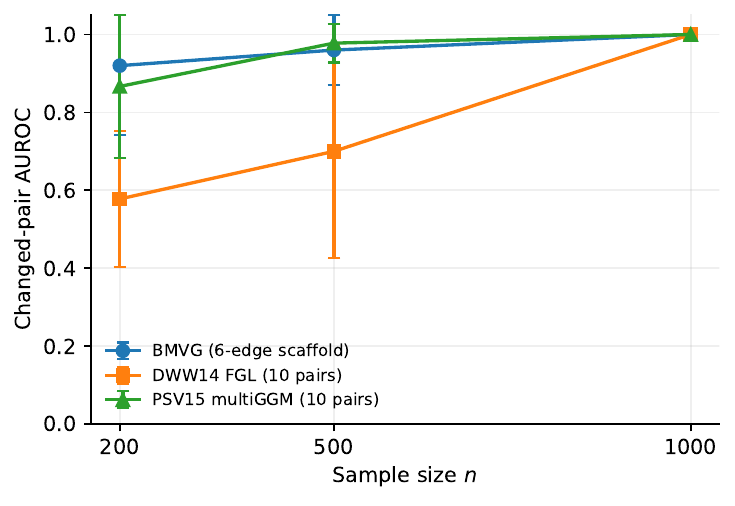}\\[-2pt]
\textbf{(b)} Orientation perturbation
\end{minipage}
\caption{BMVG searches its six-edge scaffold, whereas FGL and PSV15 search all
ten module pairs. (a) For sparse perturbations, BMVG and PSV15 are already
near-perfect at $n=200$, while FGL improves with sample size.
(b) Orientation changes are more difficult at small $n$, particularly for
FGL, but all three methods approach perfect localization by $n=1000$.
Error bars show one sample standard deviation across five independent
replicates.}
\label{fig:native_localization}
\end{figure*}

PSV15 provides a posterior probability of support change for each scalar
precision edge. In contrast, the known-truth structural-resolution study
plants a reconfiguration of an existing \(3\times3\) module interaction:
its magnitude and/or orientation changes while much of the scalar
conditional-independence support can remain present. Consequently,
posterior-mean block change
\(\|\widehat{\Delta Q}_{ij}\|_F\) localizes the planted module pair very well,
whereas the native differential-edge PIP is weak (\Cref{tab:native_localization}). The two
quantities address different inferential targets: PSV15 asks whether scalar
conditional-independence support changes, whereas BMVG characterizes how an
existing multivariate edge interaction changes in magnitude and direction.

The reported PSV15 uses four chains. Across all \(120\) conditions,
the mean maximum pairwise disagreement between chain-specific posterior-mean
precision estimates is approximately \(1.94\%\) for \(Q_0\) and \(1.90\%\)
for \(Q_1\); the worst observed disagreements are about \(4.09\%\) and
\(4.70\%\), respectively. These values are reported only as a chain-mean
agreement diagnostic, not as a substitute for \(\widehat R\) or ESS. No
planted truth or held-out test data are used in the PSV15 fit.

The three competitors expose complementary aspects of the problem. FGL can
be effective for estimating a global differential precision, PSV15 is a
strong Bayesian sparse-GGM estimator and changed-block localizer, and CPC
provides a natural common-eigenspace baseline. BMVG is distinguished by
combining context-specific Bayesian precision inference with an explicit
edge-level deformation geometry. A changed interaction can therefore be
localized and decomposed into deformation magnitude and direction, rather
than being represented only through scalar support changes or entrywise
precision differences.
\section{Additional Meteostat Results}
\label{app:meteostat}

The labels $R0,\ldots,R3$ denote four prespecified, data-driven weather
regimes fitted using training data only.  For each 3-hour anomaly field,
features summarize spatial means, standard deviations, ranges, 12-hour
tendencies, maximum wind anomaly, minimum pressure anomaly, and the norm of
the tendency vector.  They are averaged over a trailing 12-hour window using
only current and past observations, standardized using training-period
statistics, and clustered by $k$-means with $K=4$.  Labels are ordered by
increasing spatial-mean temperature anomaly at the cluster center and do not
represent predefined meteorological categories.

\subsection{Regime and seasonal separation}

Across regime pairs, mean edgewise AIRM separation ranges from $2.453$ to
$3.011$, with mean $2.800$; across seasons it ranges from $2.347$ to
$4.168$, with mean $3.511$.

\begin{table}[htbp]
\centering
\small

\begin{subtable}[t]{0.47\linewidth}
\centering
\begin{tabular}{lr}
\toprule
Context pair & Mean edge AIRM\\
\midrule
R1--R2 & 3.011\\
R2--R3 & 2.917\\
R0--R1 & 2.847\\
R1--R3 & 2.795\\
R0--R2 & 2.780\\
R0--R3 & 2.453\\
\bottomrule
\end{tabular}
\caption{Regime pairs.}
\end{subtable}
\hfill
\begin{subtable}[t]{0.47\linewidth}
\centering
\begin{tabular}{lr}
\toprule
Context pair & Mean edge AIRM\\
\midrule
Autumn--Spring & 4.168\\
Spring--Winter & 3.949\\
Summer--Winter & 3.923\\
Autumn--Summer & 3.609\\
Autumn--Winter & 3.067\\
Spring--Summer & 2.347\\
\bottomrule
\end{tabular}
\caption{Season pairs.}
\end{subtable}

\caption{
Mean AIRM distance between corresponding posterior-mean edge matrices,
averaged over the 16 fixed geographic edges.
}
\label{tab:meteostat-pairs}
\end{table}

\subsection{Global variable structure}

The leading ordinary mode remains pressure-dominated under raw, anomaly, and
innovation representations.  Relative to anomaly geometry, innovation has
larger conditional variance and effective resistance for temperature,
relative humidity, and wind speed, whereas pressure changes much less.

\begin{table}[htbp]
\centering
\small

\begin{subtable}[t]{0.55\linewidth}
\centering
\begin{tabular}{lrrrrr}
\toprule
Geometry & Lead share & Temp. & RH & Wind & Press.\\
\midrule
Anomaly    & 0.848 & 0.066 & 0.262 & 0.380 & 0.007\\
Innovation & 0.892 & 0.125 & 0.436 & 0.687 & 0.010\\
Raw        & 0.846 & 0.072 & 0.276 & 0.414 & 0.008\\
\bottomrule
\end{tabular}
\caption{Leading-mode share and conditional variance.}
\end{subtable}
\hfill
\begin{subtable}[t]{0.40\linewidth}
\centering
\begin{tabular}{lrr}
\toprule
Variable & Anomaly & Innovation\\
\midrule
Pressure          & 0.018 & 0.024\\
Temperature       & 0.154 & 0.285\\
Relative humidity & 0.549 & 0.897\\
Wind speed        & 0.786 & 1.400\\
\bottomrule
\end{tabular}
\caption{Effective resistance.}
\end{subtable}

\caption{
Complementary summaries of global variable structure. Conditional variances
are averaged across stations; effective resistances are averaged over the
fixed network.
}
\label{tab:meteostat-global}
\end{table}

Posterior partial correlations give the same qualitative pattern:
temperature coupling weakens under innovation while pressure associations
remain comparatively strong.

\subsection{Edgewise reconfiguration}

The anomaly-innovation change is spatially heterogeneous.  The largest
posterior-mean AIRM distance occurs on Half Moon Bay-San Carlos, with
several other large changes involving coastal or peninsula connections.

\begin{table}[htbp]
\centering
\small
\begin{tabular}{lrrrr}
\toprule
Edge &
$D_{\rm mean}$ &
Draw mean &
$q_{.05}$ &
$q_{.95}$\\
\midrule
Half Moon Bay--San Carlos       & 5.351 & 6.850 & 5.454 & 8.740\\
San Francisco--Half Moon Bay    & 4.609 & 5.819 & 4.461 & 7.497\\
Oakland--Novato                 & 4.464 & 5.943 & 4.437 & 8.009\\
Palo Alto--San Carlos           & 4.335 & 5.089 & 3.802 & 6.901\\
Oakland--Concord                & 4.192 & 5.647 & 4.237 & 7.175\\
Oakland--Half Moon Bay          & 3.890 & 4.744 & 3.510 & 6.314\\
\bottomrule
\end{tabular}
\caption{
Largest anomaly--innovation edge changes.
$D_{\rm mean}$ is the AIRM distance between posterior-mean edge matrices; the remaining columns summarize paired posterior draws.
}
\label{tab:meteostat-edge-change}
\end{table}

The generalized-direction analysis in \Cref{fig:meteostat_anisotropic_reconfiguration} provides the
corresponding directional interpretation: the strongest sign-certain changes
are negative and are dominated primarily by wind speed and relative humidity.

\subsection{Atmospheric distribution-shift detection}
\label{app:ood}

Six prespecified daily summaries are compared on validation data.  Selection
uses validation TPR subject to
$\mathrm{FPR}_{\rm val}\le0.05$; final-test metrics are reported only after
the statistic is frozen.

\begin{table}[htbp]
\centering
\scriptsize
\setlength{\tabcolsep}{3.5pt}
\begin{tabular}{lccccc}
\toprule
Statistic
& Val.\ AUROC
& Val.\ AUPRC
& Val.\ TPR
& Test AUROC
& Test AUPRC \\
\midrule
Maximum        & .703 & .697 & .250 & .751 & .794 \\
Mean           & .841 & .827 & .438 & .823 & .840 \\
Median         & .891 & .861 & .563 & .787 & .810 \\
\textbf{Trimmed mean}
               & .888 & .883 & .625 & .835 & .848 \\
75th percentile& .800 & .794 & .438 & .799 & .834 \\
Top-two mean   & .753 & .773 & .438 & .777 & .813 \\
\bottomrule
\end{tabular}
\caption{
Candidate daily graph-structural summaries.  The bold row indicates the
statistic selected on validation.
}
\label{tab:ood_statistic_comparison}
\end{table}

The trimmed mean achieves the largest validation TPR under the prespecified
FPR constraint and the largest validation AUPRC, and gives final-test
AUROC $0.835$ and AUPRC $0.848$.  At the primary $5\%$ validation-ID FPR
budget, test TPR is $0.36$ and empirical test FPR is $0$.  Increasing the
budget to $0.15$ raises test TPR to $0.52$ with FPR $0.033$, while at
$0.20$ the additional gain is small relative to the rise in false alarms.
The aggregate performance also varies temporally, with stronger structural
departure in December than in November; we therefore interpret this analysis
as structural OOD rather than treating every feature-space novelty event as a
graph anomaly.

\section{Additional TCGA-BRCA Results}
\label{app:tcga}

\subsection{Data and frozen scalar edge maps}

The analysis uses 1097 unique primary TCGA-BRCA tumors and the 15 genes
listed in \Cref{sec:tcga}.  The source uses the PKM2 field for PKM, so the
loader applies PKM$\leftarrow$PKM2.  Gene-wise means and standard deviations
are computed once on the full primary cohort and reused for the population
comparisons.

For each of the ten module pairs, the leading regularized-CCA coordinate is
estimated from pooled expression before ER labels are loaded, using
regularization $0.05$.  The resulting projected correlations range from
$0.172$ to $0.567$, with mean $0.403$.

\begin{figure}[htbp]
\centering
\includegraphics[
    width=.55\linewidth,
    keepaspectratio
]{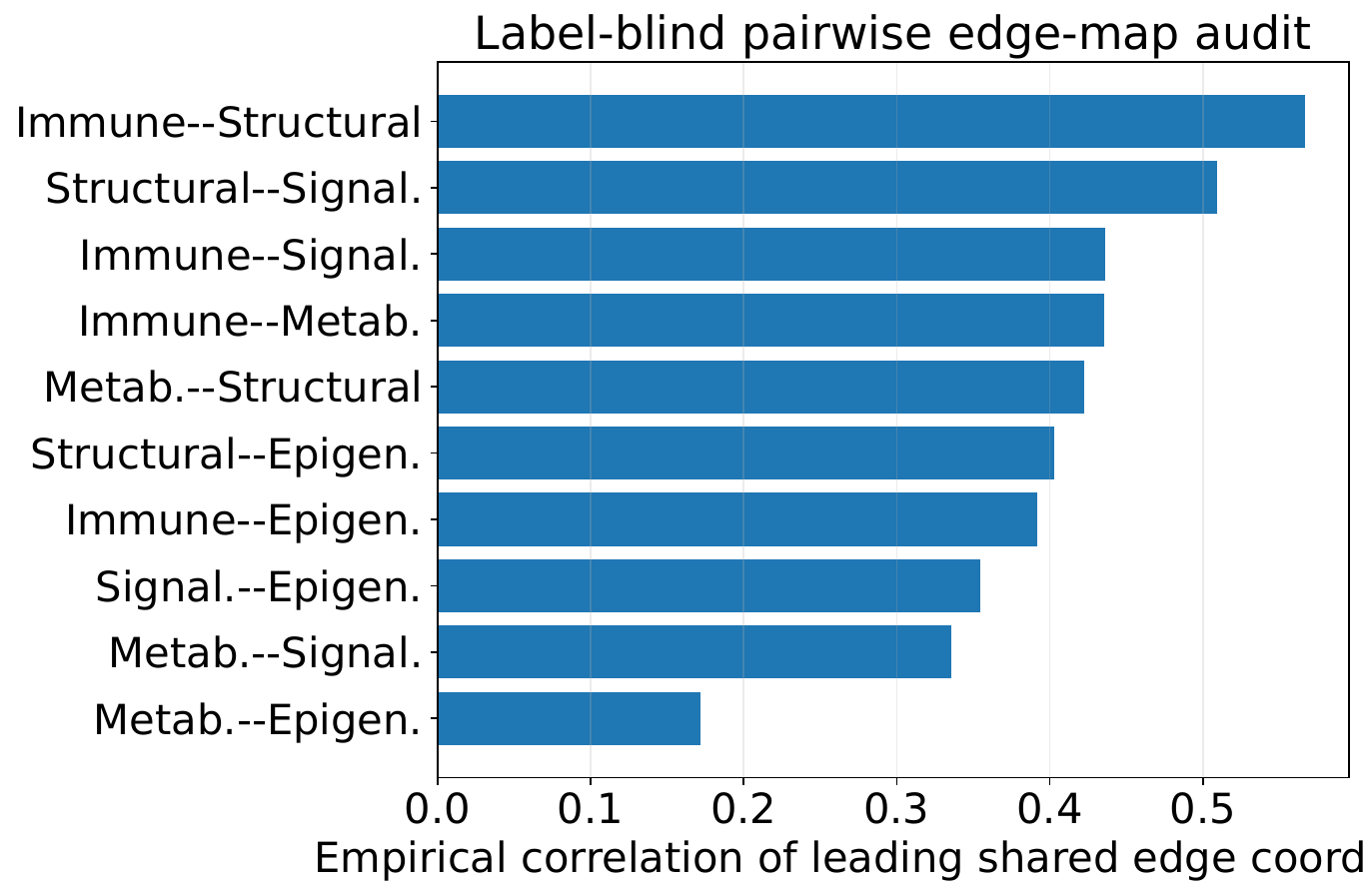}
\caption{
Empirical correlations between projected endpoint scores for the ten
module-pair coordinates.  The maps are learned without ER labels and frozen
before posterior comparison.
}
\label{fig:tcga-map-audit}
\end{figure}

\subsection{$K_5$ comparison and robustness}

The complete scaffold contains all ten module pairs.  Four chains are run for
each ER group; the largest split-$\widehat R$ is $1.000255$, the smallest
ESS is $11042.98$, and no numerical failures occur.

\begin{table}[htbp]
\centering
\scriptsize
\begin{tabular}{r l r r r r c}
\toprule
Rank & Module pair & $\bar w_e^{-}$ & $\bar w_e^{+}$ &
$D_e$ & 90\% interval for $\eta_e$ & Sign-certain\\
\midrule
1  & Structural/EMT--Signalling & 0.182 & 0.924 & 1.625 & $[1.173,2.314]$ & Yes\\
2  & Metabolic--Epigenetic      & 0.036 & 0.113 & 1.151 & $[0.014,3.142]$ & Yes\\
3  & Metabolic--Signalling      & 0.266 & 0.620 & 0.845 & $[0.541,1.211]$ & Yes\\
4  & Immune--Structural/EMT     & 0.544 & 0.983 & 0.591 & $[0.300,0.928]$ & Yes\\
5  & Signalling--Epigenetic     & 0.351 & 0.523 & 0.398 & $[0.069,0.790]$ & Yes\\
6  & Metabolic--Structural/EMT  & 0.508 & 0.667 & 0.272 & $[0.003,0.572]$ & Yes\\
7  & Structural/EMT--Epigenetic & 0.537 & 0.669 & 0.221 & $[-0.054,0.532]$ & No\\
8  & Immune--Epigenetic         & 0.325 & 0.392 & 0.189 & $[-0.222,0.715]$ & No\\
9  & Immune--Metabolic          & 0.580 & 0.618 & 0.063 & $[-0.210,0.364]$ & No\\
10 & Immune--Signalling         & 0.507 & 0.513 & 0.013 & $[-0.316,0.364]$ & No\\
\bottomrule
\end{tabular}
\caption{
ER$-$ versus ER$+$ comparison on $K_5$.  $D_e$ is the AIRM distance
between posterior-mean scalar weights; intervals are paired-draw
$5$-$95\%$ intervals for $\eta_e$.
}
\label{tab:tcga-k5-edgewise}
\end{table}

The leading result is stable to both robustness checks.  Across all 12
labeled five-cycle scaffolds, Structural/EMT-Signalling ranks first in all
six cycles containing it, while the $K_5$ and cycle-average edge distances
have Pearson $r=0.990$ and Spearman $\rho=0.988$.  Removing ER-specific
means reduces the mean edge AIRM only from $0.536816$ to $0.505015$, retaining
$94.08\%$ of the primary separation; all ten edge ranks are preserved
(Spearman $\rho=1.000$; Pearson $r=0.9934$).

\end{document}